\documentclass[lettersize,journal]{IEEEtran}
\usepackage{amsmath,amsfonts}
\usepackage{array}
\usepackage{textcomp}
\usepackage{stfloats}
\usepackage{url}
\usepackage{verbatim}
\usepackage{graphicx}
\usepackage{cite}
\usepackage[utf8]{inputenc}
\usepackage[dvipsnames]{xcolor}
\usepackage{xspace}
\usepackage{amsmath}
\usepackage{amssymb}
\usepackage{amsthm}
\usepackage{stmaryrd}
\usepackage[bb=boondox]{mathalfa}
\usepackage[ruled,vlined]{algorithm2e}
\usepackage{mathtools}
\usepackage[shortlabels]{enumitem}
\usepackage{tikz}
\usetikzlibrary{calc}    
\usetikzlibrary{positioning}
\usetikzlibrary{arrows.meta}
\usepackage{caption}
\usepackage{subcaption}
\usepackage{adjustbox}
\usepackage{float}
\usepackage{multirow}
\usepackage{todonotes}
\newtheorem{theorem}{Theorem}[section]

\newtheorem{lemma}[theorem]{Lemma}
\theoremstyle{definition}
\newtheorem{definition}{Definition}[section]
\usepackage[switch]{lineno}

\newcommand{\BB}{\textsc{BB}\xspace}
\newcommand{\CFG}{\textsc{CFG}\xspace}
\newcommand{\DAG}{\textsc{DAG}\xspace}

\newcommand{\XDD}{\textsc{XDD}\xspace}
\newcommand{\XG}{\textsc{XG}\xspace}

\newcommand{\N}{\mathbb{N}\xspace}
\newcommand{\Z}{\mathbb{Z}\xspace}

\newcommand{\inst}{I}
\newcommand{\Insts}{\mathcal{I}\xspace}

\newcommand{\Stages}{\mathcal{P}\xspace}
\newcommand{\stage}{s\xspace}

\newcommand{\D}{\mathcal{D}\xspace}

\newcommand{\Events}{\ensuremath{\mathcal{E}}\xspace}

\newcommand{\Gxg}{\ensuremath{G_{XG}}\xspace}
\newcommand{\Vxg}{\ensuremath{V_{XG}}\xspace}
\newcommand{\Exg}{\ensuremath{E_{XG}}\xspace}

\newcommand{\vxg}{\ensuremath{v}\xspace}

\newcommand{\wxg}{\ensuremath{w}\xspace}

\newcommand{\PipeStates}{\ensuremath{\mathcal{S}}}

\newcommand{\LEAF}{\ensuremath{\text{\textsc{leaf}}}\xspace}
\newcommand{\NODE}{\ensuremath{\text{\textsc{node}}}\xspace}
\newcommand{\SF}{\ensuremath{f}\xspace}

\newcommand{\Zero}{\mathbb{0}}
\newcommand{\One}{\mathbb{1}}

\newcommand{\leftME}{\ensuremath{\blacktriangleleft_{ME}}\xspace}
\newcommand{\leftFE}{\ensuremath{\blacktriangleleft_{FE}}\xspace}

\newcommand{\rhohat}{\hat{\rho}}

\usepackage{todonotes}
\newif\ifwithnotes
\withnotestrue 
\newcommand{\zhen}[1]{\ifwithnotes\todo[color=green!25!white]{#1}\fi}

\newcommand{\tho}[1]{\ifwithnotes\todo[color=red!25!white]{#1}\fi}

\begin{document}
\nolinenumbers

\title{Computing Execution Times with eXecution Decision Diagrams in the Presence \\  of Out-Of-Order Resources}

\author{
    \IEEEauthorblockN{Zhenyu Bai},
    \IEEEauthorblockN{Hugues Cassé},
    \IEEEauthorblockN{Thomas Carle},
    \IEEEauthorblockN{Christine Rochange}
    \IEEEauthorblockA{ \\
        IRIT, Université de Toulouse, CNRS, Toulouse INP, UT3 \\
        118 route de Narbonne 31062 Toulouse, France \\
        Email: \textit{firstname.surname}@irit.fr
    }
}

\markboth{Journal of \LaTeX\ Class Files,~Vol.~14, No.~8, August~2021}%
{Shell \MakeLowercase{\textit{et al.}}: A Sample Article Using IEEEtran.cls for IEEE Journals}


\maketitle
\begin{abstract}
Worst-Case Execution Time (WCET) is a key component for the verification of critical real-time applications.
Yet, even the simplest microprocessors implement pipelines with concurrently-accessed resources, such as the memory bus shared by fetch and memory stages.  Although their \emph{in-order} pipelines are, by nature, very deterministic, the bus can cause out-of-order accesses to the memory and, therefore, \emph{timing anomalies}: local timing effects that can have global effects but that cannot be easily composed to estimate the global WCET. To cope with this situation, WCET analyses have to generate important over-estimations in order to preserve safety of the computed times or have to explicitly track all possible executions. In the latter case, the presence of out-of-order behavior leads to a combinatorial blowup of the number of pipeline states for which efficient state abstractions are difficult to design. This paper proposes instead a compact and exact representation of the timings in the pipeline, using eXecution Decision Diagram (\XDD)\cite{bai_improving_2020}.  We show how \XDD can be used to model pipeline states all along the execution paths by leveraging the algebraic properties of \XDD. This computational model allows to compute the exact temporal behavior at \emph{control flow graph} level and is amenable to efficiently and precisely support WCET calculation in presence of out-of-order bus accesses. This model is finally experimented on the TACLe benchmark suite and we observe good performance making this approach appropriate for industrial applications.
\end{abstract}


\section{Introduction}
The correct behavior of hard real-time systems depends not only on its functional behavior but also on its temporal behavior. The latter is guaranteed by the scheduling analysis of the tasks composing the system, which relies on the estimation of their Worst-Case Execution Time (WCET).

With modern processors, the execution time of a code snippet is difficult to determine. For instance, on a processor equipped with cache memories, the latency of memory accesses is variable: it depends on whether the access results in a \emph{Cache Miss} or a \emph{Cache Hit}. Although the latency of memory access it-self is statically known (i.e. the latency in case of a Miss), it cannot be easily accounted when computing the execution time. 
The presence, in modern micro-architectures, of pipelined and superscalar execution and other mechanisms to favor instruction-level parallelism achieving high performance causes a large variability in the execution times and makes the WCET analysis to suffer from \emph{timing anomalies}\cite{lundqvist_timing_1999, eisinger_automatic_2006, gebhard2010timing, cassez-timing-anomalies}. Briefly, timing anomalies state that the WCET analysis cannot assert a local worst case with a constant worst case temporal effect. Illustrated on the case of cache accesses, this means that Cache Miss (longer access) cannot be asserted to be the global worst case and no constant worst time contribution to the global WCET can be determined~\cite{wenzel_principles_2005,reineke_definition_2006}. 

Unless the target processor is proved to be timing-anomalies free, a safe and precise WCET analysis has to capture them by precisely tracking the execution states of the micro-architecture. Hence the WCET computation is generally decomposed into two parts~\cite{wilhelm_static_2010}. Firstly, \emph{global analyses}, independent from pipeline's structure, are performed: they typically encompass cache and branch-prediction analyses, which determine the behavior of these mechanisms at instruction level. Secondly, the \emph{pipeline analysis} uses the information provided by \emph{global analyses} to determine how instructions are executed through the pipeline, and to determine the (worst-case) execution times of instruction sequences.

In~\cite{bai_improving_2020}, Bai et al. show that the eXecution Decision Diagram (\XDD) is good data structure to record times in pipeline analysis. An \XDD can be deemed as a lossless compression of the relationship between the execution time of an instruction sequence and the combination of the occurrence of timing variations. By implanting \XDD into the pipeline model based on the Execution Graph~(\XG), they have achieved exact and efficient pipeline analysis on sequentially executed instructions in in-order pipelines. 
In this~\cite{bai_improving_2020}, the pipeline analysis is designed to consider \emph{Basic Blocks} (\BB) of program independently by calculating their worst-case execution context. However, with the presence of out-of-order resources like shared buses, the conservative use of a worst-case context does not hold any more. We need to precisely track the possible execution contexts in order to evaluate how the concurrent accesses to the bus are interlaced. The pipeline analysis has to analyze the micro-architecture states on the whole of program i.e. at Control-Flow Graph~(CFG) level. 

\paragraph{Contribution}
This paper shows (a) how we adapt the original graph based pipeline model proposed in~\cite{bai_improving_2020} into a data-flow analysis applied at CFG level which computes exactly all possible temporal pipeline states; (b) how, by leveraging the algebraic properties of \XDD, we construct an efficient computational model of our analysis; and (c) how we exploit the precise pipeline states produced by this computational model to support a typical out-of-order resource: the shared memory bus between instruction cache and data cache to access the main memory.

\paragraph{Outline}
Section~\ref{sec:background} presents the background knowledge about the \XG model and \XDD. In section~\ref{sec:resource-based-model}, we extend the original model of \XG with \XDD to a resource-based model which is able to express the state of pipeline with a vector. Later in this section, we show how to leverage the algebraic properties of \XDD in order to ameliorate the performance of the analysis. In section~\ref{sec:analysis-on-cfg}, we show how to build the complete analysis at CFG level. Section~\ref{sec:shared-bus} extends our model to support the shared memory bus.  Experimentation in section~\ref{sec:experiments} demonstrates the efficiency of our analysis on realistic benchmarks. Several metrics are examined and discussed. Related works are presented in section~\ref{sec:related} and we conclude in Section~\ref{sec:conclusion}.



\section{Background}
\label{sec:background}

As the analyzed program has several execution paths and possibly loops, it is impossible to track all the possible execution traces.
The static WCET analysis approach we use in this paper models the whole program as a \CFG, then computes the WCET of each \BB and determines the WCET using the \emph{Implicit Path Enumeration Techniques} (IPET)~\cite{li_performance_1995}. Thus, the pipeline analysis aims to determine the execution time of each \BB, for example, using \emph{Execution Graphs}.

\subsection{Execution Graphs}
\label{sec:XG}
An eXecution Graph~(\XG)~\cite{xianfeng_li_modeling_2004,rochange_context-parameterized_2009} models the temporal behavior of an instruction sequence (like a \BB) executed in the pipeline. The key idea of \XG is to model the temporal behavior by considering the dependencies arising between instructions during their execution in the pipeline stages. 
For~example, an~instruction have to exit a pipeline stage to start its execution in the next stage, an instruction have to read a register after another instruction has written this register, etc. This results in a dependency graph: a vertex represents the progress of an instruction in a pipeline stage; the edges represent the precedence relationships between these vertices. 
Formally, let $\Insts$ be the set of machines instructions, and let ~\XG be a \emph{Directed Acyclic Graph} (\DAG) $\Gxg = \langle \Vxg, \Exg \rangle$ built for an instruction sequence $Seq \in \Insts^*$ s.t.
\begin{itemize}
    \item $\Vxg$ is the set of vertices defined by $\Vxg = \{ [\inst_i / \stage] ~|~ \inst_i \in Seq \land \stage \in \Stages \}$, with $\Stages$ the set of pipeline stages.
    \item $\Exg \subset \Vxg \times\Vxg$, the set of edges, is built according to the dependencies in the considered pipeline.
\end{itemize}

In addition, an \XG is decorated with temporal information:
\begin{itemize}
    \item $\lambda_v \in \mathbb{N}$ is the latency of the \XG vertex $v$, that is, the time spent in this vertex.
    \item $\delta_{v \to w} \in \{ 0, 1 \}$ represents the effect of dependencies of the edge $v \to w \in \Exg$. If $\delta_{v \to w} = 1$ (solid), $w$ starts after the end of $v$; if $\delta_{v \to w} = 0$ (dotted), $w$ can start at the same time or after the start of $v$.
\end{itemize}

\newcommand{\XGTikz}{
	\node at(0cm, 0cm) {FE};
	\node at(1.5cm, 0cm) {DE};
	\node at(3cm, 0cm) {EX};
	\node at(4.5cm, 0cm) {ME};
	\node at(6cm, 0cm) {WB};
	
	\def\ix{-1cm}
	\def\iw{4cm}
	\node[align=left,text width=\iw] at(\ix, -0.5cm)	{\footnotesize ($I_0$) add r3, r0, \#4};
	\node[align=left,text width=\iw] at(\ix, -1.25cm)	{\footnotesize ($I_1$) add r1, r0, r1, lsl \#2};
	\node[align=left,text width=\iw] at(\ix, -2cm)		{\footnotesize($I_2$) ldr r2, [r3]};
	\node[align=left,text width=\iw] at(\ix, -2.75cm)	{\footnotesize($I_3$) cmp r2, ip};
	\node[align=left,text width=\iw] at(\ix, -3.5cm)	{\footnotesize($I_4$) ldrgt ip, [r3]};
	\node[align=left,text width=\iw] at(\ix, -4.25cm)	{\footnotesize($I_5$) add r3, r3, \#4};
	
	\foreach \i in {0,...,5}{
		\foreach \s/\j in {FE/0,DE/1,EX/2,ME/3,WB/4}{
			\node[fill=white,draw] at(1.5cm * \j,-0.5cm - .75cm * \i) (\s\i) {};
		}
	}
	
	\foreach \i in {0,...,5}{
		\foreach \s/\t in {FE/DE,DE/EX,EX/ME,ME/WB}{
			\draw[->,blue] (\s\i) -- (\t\i);
		}
	}
	\foreach \s in {FE,DE,EX,ME,WB}{
		\foreach \i/\j in {0/1,1/2,2/3,3/4,4/5}{
			\draw[dotted,->,red, thick] (\s\i) -- (\s\j);
		}
	}
	
	\foreach \i/\j in {0/2,1/3,2/4,3/5}{
		\foreach \s/\t in {DE/FE,EX/DE,ME/EX,WB/ME}{
			\draw[dotted,->,black, thick] (\s\i) -- (\t\j);
		}
	}
	
	\foreach \i/\j in {0/2,1/3,2/4,3/5}{
	    \foreach \s/\t in {FE/FE,DE/DE,EX/EX,ME/ME,WB/WB}{
			\draw[->,black,thick] (\s\i)edge[bend right] (\t\j);
		}
	}
	
	\draw[->,ForestGreen, thick] (ME2) -- (EX3);
	\draw[->,ForestGreen, thick] (EX0) -- (ME2);
	\draw[->,orange] (FE3)--(FE4);
}

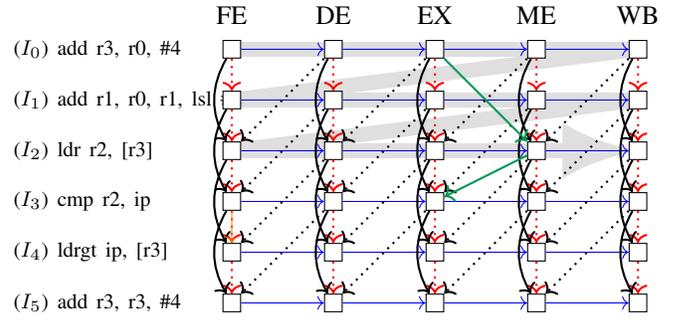
\begin{figure}[ht]
		\begin{tikzpicture}[scale=0.9]
		\draw[->, lightgray!50!white, line width=2mm, -{Latex}](0,-0.5)--(6,-0.5)--(-0,-1.25)--(6,-1.25)--(-0,-2)--(6,-2);
		\XGTikz
		\end{tikzpicture}
		\caption{\XG model of an instruction sequence.}
		\label{fig:basic-xg}
\end{figure}

Examples in this paper
consider a 5-stages (\textsc{FE} -- fetch, \textsc{DE} -- decode, \textsc{EX} -- execute, \textsc{ME} -- memory, \textsc{WB} -- write-back), in-order, 2-scalar pipeline but the presented algorithms are not limited to this configuration. Figure~\ref{fig:basic-xg} shows the \XG for this pipeline and for the sequence of instructions on the left. The~vertices correspond to the use of a stage (column headers) by an instruction (row headers) through the pipeline. The edges are generated according to the following dependencies:

\begin{itemize}
    \item The horizontal solid edges model the \emph{Pipeline Order}: an instruction goes through the pipeline in the order of stages.
    \item The vertical dotted edges model the parallel execution of instructions in the super-scalar stages (\emph{Program Order}).
    \item The vertical solid bent edges model the capacity limit of the stages -- 2 instructions per cycle (\emph{Capacity Order}).
    \item The slanted dotted edges model the capacity of FIFO queues (2 instructions) between stages (\emph{Queue Capacity}).
    \item The slanted solid edges model the \emph{Data Dependencies} between instructions when an instruction reads a register written by a prior instruction. 
\end{itemize}
The set of dependency edges shown above are typical for in-order pipelines. Depending on a particular pipeline design, rules to build the edges may be added or removed to account for specific features.

Using an \XG, the start time of an instruction in a stage $\rho_{\wxg}$ is computed as the earliest time at which all incoming dependencies are satisfied and the end time $\rho^*_{\wxg}$ as $\rho_{\wxg}$ increased by the time passed by the instruction in the stage:
\begin{align}
    \label{eq:xg-computation}
    \rho_{\wxg} & = ~\max_{\vxg \to \wxg \in E_{\XG}} \rho_{\vxg} + \delta_{\vxg \to \wxg} \times \lambda_{\vxg} \\    
    \rho^*_{\wxg} &= ~\rho_{\wxg} + \lambda_{\vxg}
\end{align}

The execution time of the instruction sequence is obtained by calculating the start time of each vertex following a~topological order in the \XG. Since the pipeline is in-order (all resources are allocated in program order), the instruction timing only depends on prior instructions meaning that, at least, one topological order exists. In in-order processors, this order is implied by the combination of the \emph{Pipeline Order}~(horizontal edges) and of the \emph{Capacity Order} (vertical edges). It is highlighted in the example \XG of Figure~\ref{fig:basic-xg} by the light gray arrow in the background.

The computation of time in \XG is fast and efficient but, as soon as the pipeline produces variable times (like Cache Hits/Misses), to be precise, the computation has to be performed for each combination of these times, causing a computation complexity blowup. Yet, the data structure presented in the next section alleviates this issue.

\subsection{Execution Decision Diagram}
\label{sec:xdd}
The eXecution Decision Diagram (\XDD) is inspired from the Binary Decision Diagram (BDD)~\cite{bdd-akers,bdd-bryant} and its Multi-Terminal BDD (MTBDD) variant~\cite{fujita1997multi}. An~\XDD is a \DAG that is recursively defined as:

\begin{definition}
\begin{equation}
	\XDD = \LEAF(k) ~|~ \NODE(e, \overline{\SF}, \SF)
\end{equation}
The Boolean variables $e \in \Events$ in the nodes are called \emph{events}: they model the uncertainty in the analysis regarding the micro-architecture state and its impact on the time e.g. whether a particular cache access results in a Hit or in a Miss. The~sub-trees $\SF, \overline{\SF} \in \XDD$ represent, respectively, the situations where the event $e$ happens or not. The leaves of {\XDD}s stores the execution times $k\in \Z^{\#}=\Z \cup \{ +\infty, -\infty\}$.
\end{definition}

As in OBDDs~\cite{bdd-bryant}, {\XDD}s deploy hash consing techniques to guarantee the unicity of the sub-trees instances and to speed up the calculations. Thus, identical sub-trees share the same instance in memory. This compression allows the {\XDD}s to represent efficiently the relationship between the combinations of events (called \emph{configurations}) and the corresponding execution times. A configuration $\gamma \in \Gamma$ is the combination of activation or inactivation of events -- $\Gamma = \wp(\Events)$, and corresponds to a path from the root node to a leaf in the \XDD \DAG.

When the events are taken into account, the actual time for each vertex is represented by a map between the event configurations and the corresponding times (i.e. in the domain $\Gamma \rightarrow \Z^{\#}$) which size is combinatorial with respect to the number of events and thus is costly in computation time.
~{\XDD}s efficiently solve this problem by factorizing identical sub-trees. In \cite{bai_improving_2020}, it has been shown that the subtree factorization frequently occurs in realistic benchmarks which largely speeds up the analysis. Yet, an isomorphism between \XDD and  $\Gamma \rightarrow \Z^{\#}$ exists: $\XDD  \xleftrightharpoons[\beta]{\alpha}  (\Gamma \rightarrow \Z^{\#})$. This means that \XDD can be deemed as a lossless compression of the map $(\Gamma \rightarrow \Z^{\#})$. In the remainder of the paper, we note $f[\gamma]$ (for $f \in \XDD$ or $f \in \Gamma \rightarrow \Z^{\#}$) the time corresponding to the configuration $\gamma$ in $f$. 

Figure~\ref{fig:example-xdd} shows an example of \XDD with 8 possible configurations: right edges of a node $e$ represent the activation of event $e$ and the left ones the non-activation. The~original map between configuration and the execution time is shown in the figure~\ref{fig:explicit-representation} on the right  ($IC$ and $DC$ represent, respectively, the instruction and data cache events, over-line bar denotes inactive events). 
\begin{figure}[h]
    \centering
    \begin{subfigure}{0.2 \textwidth}
    \tikzstyle{xdd node}=[shape=circle, draw, inner sep=.4]
    \begin{tikzpicture}[level distance=6mm,->,baseline]
            \begin{scope}[xshift=0mm, yshift=0mm]
            \def\off{2mm}
        	\node[xdd node] (DC2) { \tiny $DC_2$ };
        	\node[xdd node, below left=\off and 2mm of DC2] (IC1) { \tiny $IC_1$ };
        	\node[xdd node, below left=\off and 2mm of IC1] (IC0a) { \tiny $IC_0$ };
        	\node[right=4mm of IC0a] (L24) { $24$ };
        	\node[xdd node, right=12mm of IC0a] (IC0c) { \tiny $IC_0$ };
        	\node[below left=\off and 0mm of IC0a] (L6) { $7$ };
        	\node[below right=\off and 0mm of IC0a] (L15) { $16$ };
        	\node[below right=\off and \off of IC0c] (L25) { $25$ };
        	\draw[->,dashed] (DC2) edge (IC1);
        	\draw[->] (DC2) edge (IC0c);
        	\draw[->,dashed] (IC1) edge (IC0a);
        	\draw[->,] (IC1) edge (L24);
        	\draw[->,dashed] (IC0a) edge (L6);
        	\draw[->] (IC0a) edge (L15);
        	\draw[->] (IC0c) edge (L25);
        	\draw[->,dashed] (IC0c) edge[bend left=60,out=60] (L15);
        	\end{scope}
        \end{tikzpicture}
        \caption{Example XDD}
        \label{fig:example-xdd}
    \end{subfigure}
    \hfill
    \begin{subfigure}{0.2\textwidth}
    \footnotesize
    \begin{tabular}{|c|c|}
    \hline
        Configuration       & time \\
                \hline
        $IC_0, IC_1, DC_2$  & 25  \\
        $IC_0, \overline{IC_1}, DC_2$  & 25  \\
        $\overline{IC_0}, IC_1, DC_2$  & 16  \\
        $\overline{IC_0}, \overline{IC_1}, DC_2$  & 16  \\
        $IC_0, IC_1, \overline{DC_2}$  & 24  \\
        $\overline{IC_0}, IC_1, \overline{DC_2}$  & 24 \\
        $IC_0, \overline{IC_1}, \overline{DC_2}$  & 16 \\
        $\overline{IC_0}, \overline{IC_1}, \overline{DC_2}$  & 7 \\
    \hline
    \end{tabular}
    \caption{Explicit representation}
    \label{fig:explicit-representation}
    \end{subfigure}
    \begin{subfigure}{0.4\textwidth}
        \medskip
        \tikzstyle{xdd node}=[shape=circle, draw, inner sep=.4]
        \begin{tikzpicture}[level distance=6mm,->,baseline]
        	\def\off{2mm}
        	\begin{scope}[xshift=0mm, yshift=0mm]
                \node[xdd node](DC2){\tiny$DC_2$};
                \node[xdd node, below left=\off and 2mm of DC2] (IC1) {\tiny$IC_1$};
                \node[xdd node, below right=\off and 2mm of DC2] (IC0) {\tiny$IC_0$};
                
                \node[below left=\off and -1mm of IC1] (L3){ $3$};
                \node[below right=\off and -1mm of IC1] (L5){ $5$ };
                
                \node[below left=\off and -1mm of IC0] (L4){ $4$};
                \node[below right=\off and -1mm of IC0] (L6){ $6$ };
                
                \draw[->, dashed] (DC2) edge (IC1);
                \draw[->] (DC2) edge (IC0);
                \draw[->, dashed] (IC1) edge (L3);
                \draw[->] (IC1) edge (L5);
                \draw[->, dashed] (IC0) edge (L4);
                \draw[->] (IC0) edge (L6);
                
                \node[right=3mm of IC0](oplus){\Large $\oplus$  };
        	\end{scope}
        	
        	\begin{scope}[xshift=25mm, yshift=-3mm]
                \node[xdd node](DC2){\tiny$DC_2$};
                \node[below left=\off+2mm and -1mm of DC2] (L4){ $4$};
                \node[below right=\off+2mm and -1mm of DC2] (L7){ $7$ };
                \node[right=12mm of oplus](equals){\large $=$  };
                \draw[->, dashed] (DC2) edge (L4);
                \draw[->] (DC2) edge (L7);
        	\end{scope}
        	
        	\begin{scope}[xshift=48mm, yshift=-1mm]
                \node[xdd node](DC2){\tiny$DC_2$};
                \node[xdd node, below left=\off and 2mm of DC2] (IC1) {\tiny$IC_1$};
                
                \node[below left=\off and -1mm of IC1] (L4){ $4$};
                \node[below right=\off and -1mm of IC1] (L5){ $5$ };
                
                \node[below right=\off and 0mm of DC2] (L7){ $7$};

                \draw[->, dashed] (DC2) edge (IC1);
                \draw[->] (DC2) edge (L7);
                \draw[->, dashed] (IC1) edge (L4);
                \draw[->] (IC1) edge (L5);
        	\end{scope}
        \end{tikzpicture}
        \caption{Example of $\oplus$}
        \label{fig:example-oplus}
    \end{subfigure}
    \vspace{2mm}
    \caption{Example of {\XDD}s}
    \label{fig:xdd}
\end{figure}

\cite{bai_improving_2020} has shown that any binary operation $\odot$ on $\Z^{\#}$ used in $(\Gamma \rightarrow \Z^\#)$ can be transferred in the \XDD domain in an equivalent operation $\boxdot$ such that performing the operation in \XDD domain is lossless:


\begin{equation}
    \begin{split}
        & \forall s_1, s_2 \in (\Gamma \rightarrow \Z^\#)^2, \forall \gamma \in \Gamma,\\
        & s_1[\gamma] \odot s_2[\gamma] = (\alpha(s_1) \boxdot \alpha(s_2))[\gamma] \\
    \end{split}
\label{eq:binop-isomorphism}
\end{equation}
With $\alpha$ the morphism from $(\Gamma \to \mathbb{Z}^\#)$ to \XDD.


\medskip
The implementation of this operation is detailed in~\cite{bai_improving_2020}. Shortly, $\boxdot$ combines the \XDD operands along the sub-trees and applies $\odot$ when leaves need to be combined.
As the operations used in \XG analysis are $max$ and $+$ (Eq.~\ref{eq:xg-computation}), the equivalent operations on {\XDD}s are, respectively $\oplus$ and $\otimes$:

\begin{equation}
    \begin{split}
        & \forall s_1, s_2 \in (\Gamma \rightarrow \Z^\#)^2, \forall \gamma \in \Gamma,\\
        & s_1[\gamma] + s_2[\gamma] = (\alpha(s_1) \otimes \alpha(s_2))[\gamma] \\
        & max(s_1[\gamma], s_2[\gamma]) = (\alpha(s_1) \oplus \alpha(s_2))[\gamma] 
    \end{split}
\end{equation}

The work of $\oplus$ and its ability to reduce the size of {\XDD}s is illustrated in the example of Figure~\ref{fig:example-oplus}. 

The isomorphism
guarantees that using \XDD to perform the \XG analysis is \emph{precise} by following the dependency resolution rule (Equation~\ref{eq:xg-computation}) and following the proposed topological order. 
By exploiting the properties of {\XDD}s, the next section proposes a new paradigm of \BB time calculation to support the pipeline analysis at the \CFG level and of out-of-order accesses to the bus.
\section{Resource Based Model}
\label{sec:resource-based-model}
 

The usual approach -- consisting in building and solving an {\XG} for each {\BB} on its own, is no more sustainable when out-of-order bus accesses have to be supported. Indeed, the~bus access interactions can span over \BB bounds while {\XG}s for whole execution paths are inconvenient to build and the number of execution paths is intractable.


This section proposes to solve this issue by turning the original \XG model into a state machine model where the pipeline analysis is performed by applying transitions on the pipeline states. Moreover, by leveraging the algebraic property of \XDD, we improve the computational model by implementing the transitions as matrices multiplications. The matrices can be pre-computed before the pipeline analysis.

\subsection{Temporal State}

The dependencies $d\in\D$ in the \XG model the uses and releases of \emph{resources} e.g. stages, queues etc. For instance, in our 5-stage in-order pipeline, determining the start time of an instruction $I_i \in \Insts$ in stage $DE$ requires:
\begin{itemize}
    \item the start time of the previous instruction in the DE stage: $\rho_{[I_{i-1}/DE]}$ (Program Order),
    \item the end time of the second last instruction in the DE stage: $\rho^*_{[I_{i-2}/DE]}$ (Capacity Order),
    \item the end time of $I_i$ in the FE stage: $\rho^*_{[I_i/FE]}$ (Pipeline Order),
    \item the start time of the second last instruction in the EX stage: $ \rho_{[I_{i-2}/EX]}$ (Queue Capacity).
\end{itemize}
Where $I_{i-n}$ represents the $n^{th}$ previous instruction. $\rho_{[I_i/s]}$ and $\rho^*_{[I_i/s]}$ respectively stand for the start and the end time of the $[I_i/s]$ vertex.  
The actual start time of $I_i$ in $DE$ is the earliest date at which all dependencies are satisfied:
\begin{align*}
    \rho_{[I_i/DE]}         & = \rho_{[I_{i-1}/DE]} \oplus \rho^*_{[I_{i-2}/DE]} \\ 
                            & \oplus  \rho^*_{[I_i/FE]} \oplus \rho_{[I_{i-2}/EX]} \\
    \rho^*_{[I_{i-2}/DE]}   & = \rho_{[I_{i-2}/DE]} \otimes \lambda_{[I_{i-2}/DE]} \\
    \rho^*_{[I_{i}/FE]}   & = \rho_{[I_{i}/FE]} \otimes \lambda_{[I_{i}/FE]} \\
\end{align*}
This corresponds to the computation of Equation~\ref{eq:xg-computation} extended to the \XDD domain: we use $\oplus$ instead of $max$ and $\otimes$ instead of $+$.
As in the integer case, each computation requires the results of the computations of previous instructions (e.g. $\rho_{[I_{i}/DE]}$ requires $\rho_{[I_{i-1}/DE]}$ and $\rho^*_{[I_{i-2}/DE]}$) which correspond to the release time of the concerned resources. 



Table~\ref{tab:vec-de} sums up the dependency information required to compute the start time of any instruction $I_i$ in the $DE$ stage. The necessary information may differ depending on the pipeline architecture, but an important point is that any architecture that can be described in the \XG model can also be expressed as a vector of {\XDD}s as Table~\ref{tab:vec-de}.
\begin{table*}[htbp]
    \centering
    \begin{tabular}{|c|c|c|c|c|c|}
    \hline
         Program Order & \multicolumn{2}{c|}{Capacity Order} & \multicolumn{1}{c|}{Pipeline Order} & \multicolumn{2}{c|}{Queue Capacity}  \\
         \hline
         $\rho_{[I_{i-1}/DE]}$ &  $\rho^*_{[I_{i-2}/DE]}$ & $\rho^*_{[I_{i-1}/DE]}$ & $\rho^*_{[I_i/FE]}$ & $\rho_{[I_{i-2}/EX]}$ & $\rho_{[I_{i-1}/EX]}$\\
    \hline
    \end{tabular} 
    \caption{Necessary information determining the start time of any $I_i$ in the DE stage.}
    \label{tab:vec-de}
\end{table*}

In the same fashion, such a vector can be built for each instruction and for each stage of the pipeline. Table~\ref{tab:vec-all} shows the complete dependency information to be maintained for all stages of the example pipeline: each line represents a stage and each column represents a dependency on a resource to be satisfied to start the stage execution. The symbol $-\infty$ is used when no dependency is required\footnote{ $-\infty$ is convenient as it is neutral for the $max$ operation.}. $I_{fetch}$, $I_{load}$, $I_{store}$ and $I_{R_i}$ are, respectively, the last instructions that fetched an instruction block from memory, performed a load, a store and wrote to register $R_i$ (in stage $s_{Ri}$).

\begin{table*}[htbp]
\resizebox{\textwidth}{!}{
    \centering
    \begin{tabular}{|c|c|c|c|c|c|c|}
        \hline
    &Prog. Order            &\multicolumn{2}{c|}{Capacity Order}               & Pipeline Order            &\multicolumn{2}{c|}{Queue Capacity} \\
        \hline    
FE  &$\rho_{[I_{i-1}/FE]}$ &$\rho^*_{[I_{i-1}/FE]}$&$\rho^*_{[I_{i-2}/FE]}$  &$-\infty$          &$\rho_{[I_{i-1}/DE]}$&$\rho_{[I_{i-2}/DE]}$\\
DE  &$\rho_{[I_{i-1}/DE]}$ &$\rho^*_{[I_{i-1}/DE]}$&$\rho^*_{[I_{i-2}/DE]}$  &$\rho^*_{[I_{i}/FE]}$ &$\rho_{[I_{i-1}/EX]}$&$\rho_{[I_{i-2}/EX]}$\\
EX  &$\rho_{[I_{i-1}/EX]}$ &$\rho^*_{[I_{i-1}/EX]}$&$\rho^*_{[I_{i-2}/EX]}$ &$\rho^*_{[I_{i}/DE]}$  &$\rho_{[I_{i-1}/ME]}$&$\rho_{[I_{i-2}/ME]}$\\
ME  &$\rho_{[I_{i-1}/CM]}$ &$\rho^*_{[I_{i-1}/ME]}$&$\rho^*_{[I_{i-2}/ME]}$ &$\rho^*_{[I_{i}/EX]}$  &$\rho_{[I_{i-1}/CM]}$&$\rho_{[I_{i-2}/CM]}$\\
CM  &$\rho_{[I_{i-1}/ME]}$ &$\rho^*_{[I_{i-1}/CM]}$&$\rho^*_{[I_{i-2}/CM]}$ &$\rho^*_{[I_{i}/ME]}$  &$-\infty$         &$-\infty$ \\
        \hline 
    & Fetch Order             &\multicolumn{2}{c|}{Memory Order}           &\multicolumn{3}{c|}{Data Dependencies} \\
        \hline
FE  &$\rho^*_{[I_{fetch}/FE]}$&\multicolumn{2}{c|}{$-\infty$}        &\multicolumn{3}{c|}{$-\infty$} \\
DE  &$-\infty$          &\multicolumn{2}{c|}{$-\infty$}        &\multicolumn{3}{c|}{$-\infty$} \\
EX  &$-\infty$          &\multicolumn{2}{c|}{$-\infty$}        &$\rho^*_{[I_{R0}/s_{R0}]}$&$\rho^*_{[I_{R0}/s_{R1}]}$&$...$ \\
ME  &$-\infty$          &$\rho^*_{[I_{load}/ME}$&$\rho^*_{[I_{store}/ME]}$ &$\rho^*_{[I_{R0}/s_{R0}]}$&$\rho^*_{[I_{R1}/s_{R1}]}$&$...$ \\
CM  &$-\infty$          &\multicolumn{2}{c|}{$-\infty$}        &\multicolumn{3}{c|}{$-\infty$} \\
        \hline
    \end{tabular}
    }
    \caption{The \emph{temporal state}.}
    \label{tab:vec-all}
\end{table*}

Finally, Table~\ref{tab:vec-all} sums up the set of dependencies an instruction has to satisfy considering all possible pipeline stage, i.e. $\D$. Grouped in an \emph{\XDD vector}, defined as $\PipeStates = \XDD^{|\D|}$, they precisely represent the \emph{temporal state} of the pipeline. For a given stage $s$, a \emph{temporal state} $\vec{S} \in \PipeStates$, and $\D_{[I_i/s]} \subset \D$ the set of dependencies applicable to \XG vertex $[I_i/s]$, start and end times can now be rewritten as:

\vspace{-6mm}
\begin{align}
    \rho_{[I_i/s]} & = \bigoplus_{d \in \D_{[I_i/s]}} \vec{S}[i_d] \label{eq:max-required-resource}\\
    \rho^{*}_{[I_i/s]} = &\rho_{[I_i/s]} \otimes \lambda_{[I_i/s]} \label{eq:consume-time}
\end{align}

Notice that the function $\delta_{v \to w}$ is useful as its effect is supported by the dependency on start or end time in $\vec{S}$.


\subsection{Pipeline Analysis with Temporal States}
\label{sec:pipe-analysis-with-state-vec}
We now present how \emph{temporal states} are updated during the analysis to account for the execution of instructions in the pipeline.
To simplify the computations, we add a slot $\rho$ at index $i_\rho$ into the state vector that records the \emph{current time} all along the analysis, which we call the \emph{time pointer}.

\begin{definition}
Following the principle of \XG analysis, the~behavior of an instruction in a stage can be divided into four steps.
\begin{itemize}
    \item Step 1. Before being executed in the stage, the instruction waits until all dependencies are satisfied (Eq.~\ref{eq:max-required-resource}). To~model this behavior with the \emph{temporal state}, the~time pointer is reset to $\mathbb{0} = LEAF(-\infty)$~(Step 1.1). Then, each dependency time is accumulated with $\oplus$ into the time pointer~(Step~1.2). At the end of Step~1, the time pointer records the maximum release time of all dependencies which is the actual start time for the analyzed \XG vertex. The transitions for the \emph{temporal state} are defined with the functions $\tau_{reset}$ and $\tau_{wait}$:
    
    \begin{equation}
    \begin{split}
        &\tau_{reset}: \PipeStates \rightarrow \PipeStates, \\
        &\tau_{reset}(\vec{S}) = \vec{S}'~|~
        \begin{cases}
        \vec{S}'[i] = \vec{S}[i] \otimes \Zero \text{ if } i = i_{\rho} \\ 
        \vec{S}'[i] = \vec{S}[i] \text{ otherwise }
        \end{cases}
    \end{split}
    \end{equation}
    \begin{equation}
    \begin{split}
        &\tau_{wait}: \N \times \PipeStates \rightarrow \PipeStates, \\
        &\tau_{wait}(x,\vec{S}) = \vec{S}'~|~
        \begin{cases}
        \vec{S}'[i] = \vec{S}[i] \oplus \vec{S}[x]  \text{ if } i = i_{\rho} \\ 
        \vec{S}'[i] = \vec{S}[i] \text{ otherwise }
        \end{cases}
    \end{split}
    \end{equation}

    $\tau_{wait}$ has to be called for each dependency (with index $x$ in the \emph{temporal state} vector) of the current vertex.
    
    
    \item Step 2. Some resources are released at the start of an \XG vertex. The corresponding dependencies (e.g. Program Order and Queue Capacity) have to be updated with the start time $\rho$. This is done with $\tau_{move}$:
    \begin{equation}
    \begin{split}
        &\tau_{move}: \N \times \N \times \PipeStates \rightarrow \PipeStates, \\
        &\tau_{move}(i_{dest}, i_{src} ,\vec{S}) = \vec{S}'~|~
        \begin{cases}
        \vec{S}'[i] = \vec{S}[i_{src}]\text{ if } i = i_{dest} \\ 
        \vec{S}'[i] = \vec{S}[i] \text{ otherwise }
        \end{cases}
    \end{split}
    \end{equation}    
    $\tau_{move}$ copies a vector element into another element and overwrite the destination value. Updating the dependency of single resource turns out to copy $\rho$ into the slot of the dependency: for example updating the start time of a stage $s \in Stages$ with index $i_s$ consists in  $\tau_{move}(i_{s}, i_{\rho}, \vec{S})$. The state of FIFO resources (like queues) requires to update several \emph{temporal state} slots ($i$ to $i+n-1$ with $n$ the FIFO capacity) to express the shift of the $n$ last FIFO uses. Hence FIFO resources are updated by a series of $\tau_{move}$ on the $n$ FIFO slots in the \emph{temporal state} and by setting the first slot to $\rho$, the use time for the first FIFO element:
    \begin{align*}
        &\forall j \in [i, i+n-2], \tau_{move}(j+1, j, \vec{S}); \\ 
        &\tau_{move}(i, i_{\rho}, \vec{S});
    \end{align*}

    \item Step 3. The started instruction spends $\lambda_{[I_i/s]}$ cycles in the stage. 
    \begin{equation}
    \begin{split}
        &\tau_{consume}: \N \times \PipeStates \rightarrow \PipeStates, \\
        &\tau_{consume}(\lambda_{[I_i/s]} ,\vec{S}) = \vec{S}'|
        \begin{cases}
        \vec{S}'[i] = \vec{S}[i] \otimes \lambda_{[I_i/s]} \text{ if } i = i_{\rho} \\ 
        \vec{S}'[i] = \vec{S}[i] \text{ otherwise }
        \end{cases}
    \end{split}
    \end{equation}
    
    \item Step 4. The instruction finishes its execution and the dependencies recording the end time of the current vertex are updated. The $\tau_{move}$~ operation  of Step~2  is used.
\end{itemize}
\label{def:steps-at-stage}
\end{definition}
As in the original \XG resolution model,
the computational model with \emph{temporal states} has to follow the topological order so that the times recorded in the {\XDD} vector refer to the~correct timing of resources. In other words, if the state is correctly updated according to the rules stated above, the~resource-based model is equivalent to the original \XG analysis but expressed in state machine fashion. The implementation using {\XDD}s extends the model to consider all possible cases according to the timing variations without any loss. The \BB analysis is consequently \emph{exact} with respect to the \XG pipeline model.
\subsection{The computational model}
An important property of \XDD domain is that, equipped with $\oplus$ and $\otimes$, it forms the semiring $\langle \XDD, \oplus, \otimes, \Zero, \One \rangle$ with $\Zero = LEAF(-\infty)$ and $\One = LEAF(0)$.  As the functions $\tau$ are affine in this domain, their application can be expressed as matrix multiplications. By combining and pre-computing these matrices, they will help to speed up the pipeline analysis at \CFG level as some {BB}s need to be recomputed several times in different execution contexts.

Scalar and matrix multiplication on \XDD semiring is similar to the linear algebra over $\mathbb{R}$ by replacing $+$ by $\oplus$, $\times$ by~$\otimes$:

\begin{definition}
The scalar multiplication is defined by:
\begin{multline*}
    \cdot: \XDD^{N} \times \XDD^{N} \rightarrow \XDD, \\
    [f_0, f_1, ..., f_{N-1}] \cdot [f_0', f_1', ..., f_{N-1}'] = \bigoplus_{0 \le i \le N-1} f_i \otimes f_i'\\
\end{multline*}
\end{definition}
\vspace{-5mm}
\begin{definition}
The matrix multiplication is defined by:
\begin{multline*}
    \cdot: \XDD^{N \times M} \times \XDD^{M \times L} \rightarrow \XDD^{N \times L}, \\
    B \cdot C = 
    \begin{bmatrix}
    &  & \\
    & ~ A_{i,j}& \\
    & &
    \end{bmatrix} | A_{i,j} = \bigoplus_{1 \le k \le M}B_{i,k} \otimes C_{k,j} \\
\end{multline*}
\end{definition}
\vspace{-5mm}
\begin{definition}
The identity matrix $Id$ on $\XDD$ semiring is defined by:
\begin{align*}
    Id &=
    \begin{bmatrix}
        &&\\&A_{i,j}&\\ &&
    \end{bmatrix}| A_{i,j} = \begin{cases} 
                                \One \text{ if  } i = j \\
                                \Zero \text{ otherwise}
                            \end{cases} \\
\end{align*}

\label{def:state-transitions}
\end{definition}
\vspace{-5mm}
Notice that, by definition, $\vec{S} \cdot Id = \vec{S}$: any matrix column at index $i$ composed of $\Zero$ except with a $\One$ in the row $i$ maintains unchanged the value of $\vec{S}[i]$ in the resulting vector. 
To implement the transitions functions $\tau$ as matrix multiplications, the~matrix $Id$ is taken as a basis and only the cells that have an effect on the vector have to be changed

\begin{enumerate}

    \item A $\Zero$ on the diagonal of the $Id$ matrix at the timer pointer position resets it: $\vec{S}[i_{\rho}] \otimes \Zero = \Zero$:
        \begin{equation}
        \begin{split}
            &\tau_{reset}(\vec{S}) = \vec{S} \cdot M_{reset} \\
            &= \vec{S}\cdot 
                            \begin{bmatrix}
                                &&\\&A_{i,j}&\\ &&
                            \end{bmatrix} | A_{i,j} =\begin{cases}
                                                        \Zero \text{ if } i = j = i_{\rho} \\
                                                        Id_{i,j} \text{ otherwise }
                                                    \end{cases}
        \end{split}
        \end{equation}

    \item For a given slot at index $x$ in $\vec{S}$, $\tau_{wait}(x, \vec{S})$ is represented by a matrix $M_{wait(x)}$ with a $\One$ at position $(x, i_{\rho})$ resulting in the operation $\rho \oplus (\One \otimes \vec{S}[x])$:
        \begin{equation}
        \begin{split}
            &\tau_{wait}(x, \vec{S}) = \vec{S} \cdot M_{wait(x)} \\
            &=\vec{S} \cdot
                            \begin{bmatrix}
                                &&\\&A_{ij}&\\ &&
                            \end{bmatrix} | A_{i j} =\begin{cases}
                                                        \One \text{ if } i = i_{\rho} \wedge j = x \\
                                                        Id_{ij} \text{ otherwise }
                                                    \end{cases}
        \end{split}
        \end{equation}

    \item $\tau_{move}(i_{src}, i_{dest}, \vec{S})$ is represented by a matrix $M_{move(i_{src}, i_{dest})}$ where the element at $(i_{dest}, i_{dest})$ is set to $\Zero$ and the element $(i_{dest}, i_{src})$ to $\One$ s.t. element $i_{dest}$ in the result becomes $(\Zero \otimes \vec{S}[i_{dest}]) \oplus (\One \otimes \vec{S}[i_{src}]) = \vec{S}[i_{src}]$.
    
        \begin{equation}
        \begin{split}
            &\tau_{move}(i_{src}, i_{dest}, \vec{S}) = \vec{S} \cdot M_{move(i_{src}, i_{dest})} \\
            &= \vec{S}\cdot 
                            \begin{bmatrix}
                                &&\\&A_{ij}&\\ &&
                            \end{bmatrix} | A_{ij} =\begin{cases}
                                                        \Zero \text{ if } i = j = i_{dest} \\
                                                        \One \text{ if } i = i_{src} \wedge j = i_{dest} \\
                                                        Id_{i,j} \text{ otherwise }
                                                    \end{cases}
        \end{split}
        \end{equation}

    \item For a given latency $\lambda$, $\tau_{consume}(\lambda, \vec{S})$ can be represented by a matrix $M_{consume(\lambda)}$, obtained from $Id$ by putting $\lambda$ at position $(i_{\rho}, i_{\rho})$.
        \begin{equation}
        \begin{split}
            &\tau_{consume}(\lambda, \vec{S}) = \vec{S} \cdot M_{consume}^{\lambda} \\
            &= \vec{S}\cdot 
                            \begin{bmatrix}
                                &&\\&A_{ij}&\\ &&
                            \end{bmatrix} | A_{ij} =\begin{cases}
                                                        \lambda \text{ if } i = j = i_{\rho} \\
                                                        Id_{i,j} \text{ otherwise }
                                                    \end{cases}
        \end{split}
        \end{equation}

\end{enumerate}

\begin{theorem}
Each applied transition  function $\tau$ to the timing vector is a linear map from $\PipeStates$ to $\PipeStates$
\end{theorem}
\begin{proof}
Direct since we have already given the matrix representation of each transition in Definition~\ref{def:state-transitions}.
\end{proof}





Consequently, the operation performed at each step is also linear because they are combination of $\tau$ functions. Their~matrix representation is simply the multiplication of each invoked~$\tau$. For example,
\begin{align*}
    M_{Step_1[I_i/s]}  &= M_{reset} \cdot \prod_{d \in \D_{[I_i/s]}} M_{wait(i_d)}
\end{align*}
With $\D_{[I_i/s]}$ the set of dependencies required by $[I_i/s]$ and $i_d$ the index of resource $d$ in the state vector.

Similarly, we can express $M_{Step_2[I_i/s]}, M_{Step_3[I_i/s]}$ and $M_{Step_4[I_i/s]}$ by invoking the corresponding $\tau$ functions.
As~each step is linear, the operation when analyzing one instruction on a stage is also linear because it is the combination of the 4 steps.
\begin{equation*}
    M_{[I_i/s]} = M_{Step_1[I_i/s]} \cdot M_{Step_2[I_i/s]} \cdot M_{Step_3[I_i/s]} \cdot M_{Step_4[I_i/s]}
\end{equation*}

Finally, the whole analysis of a BB $a \in V$ is composed by the analysis of each instruction on each stage:
\begin{equation*}
    M_{a} = \prod_{I_i \in a} \prod_{s \in \Stages} M_{[I_i/s]}
\end{equation*}

With a matrix as $M_a$, it is easy and fast to compute the output \emph{temporal state} $\vec{S}' \in \PipeStates$ corresponding to an input \emph{temporal state} $\vec{S} \in \PipeStates$ for a \BB $a$:
\begin{equation}
    \vec{S}' = \vec{S} \cdot M_a
\end{equation}

\section{Pipeline Analysis on the CFG}
\label{sec:analysis-on-cfg}
This section extends the \emph{temporal state} computational model, presented in the previous section, to the complete analysis of the \CFG. It consists, mainly, in tracking the explicit set of possible \emph{temporal states} for each \BB all over the \CFG execution paths.

\subsection{Computing the context with Rebasing operation} 
So far, the \emph{temporal state} contains times relative to the start of a \BB. As the analysis on \CFG starts from the entry point of the program, the recorded times are execution times relative to the start of the program and the \emph{temporal states} have to be tracked for all possible execution paths. This is generally infeasible because of the number of execution paths and especially because of the presence loops. In fact, the main reason to compute exact \emph{temporal states} at \CFG level is to determine bus accesses timings but these timings does not need to be absolute with respect to the start of the program. Instead, the times can be relative to different time bases arbitrarily chosen, while, to preserve the soudness of the computation,  {\XDD}s with different bases are not mixed. We call this operation \emph{rebasing}.


 \emph{Rebasing} a state is changing the origin of the timeline of the times it contains. For now, the \emph{temporal state} at the end of a \BB $a$ represents the delay induced by the execution of $a$ to the start of following {\BB} $b$. Considering that a new time base $T \in \XDD$ is the start of $b$, we can get a new \emph{temporal state} relative to~$T$ by subtracting $T$ from the times in the \emph{temporal state} in the base of $a$. The outcome is a \emph{temporal state} containing {\XDD}s with positive or negative times relative to~$T$. The relationship between times and events in the \emph{temporal state} is preserved. The subtraction in {\XDD}s $\oslash$ is built in the usual way from $-$ operator~(Eq.~\ref{eq:binop-isomorphism}).

\emph{Rebasing} a \emph{temporal state} is lossless simply because $\oslash$ is reversible. By adding $T$ (with $\oplus$), one can find back the state before rebasing.
Rebasing is very helpful to reduce the size of {\XDD}s in the \emph{temporal state}: an event removed by rebasing has no effect on the following {\BB}s but it does not mean it has no effect. In fact, its contribution to the overall WCET is simply linear with respect to the number of occurrences of the {\BB}.
Intuitively, the execution of an instruction
depends on the execution of nearby instructions and thus, the effect of events is rather short term and it is often eliminated by rebasing.




\subsection{Events Generation within loops}
The events calculated by global analyses are linked to a particular instruction.
The pipeline analysis of a \BB presented so far deems the occurrence of events unique. This is not true when an event arises in a \BB contained in a loop as it may occur or not in different iterations. We would get unsound timings if we denotes these different event occurrence with the same event node in the \XDD. To fix this, a \emph{generation number} is associated with each event. To prevent \emph{temporal state} blowup, this \emph{generation number} is relative to the current iteration and is incremented in the current \emph{temporal state} each time the analysis restarts the loop. The \emph{generation number} thus distinguishes the events in different iterations. However, this method does not result in an endless increase of generation because (a) the effect of events is often bound in the time and (b) the WCET calculation requires to bound the loop iterations.

\subsection{The CFG pipeline analysis}
Finally, the complete pipeline analysis is designed like a~classical data-flow analysis with a work list. Each \BB is associated with a set of input \emph{temporal states} and a set of output \emph{temporal states} (initially empty). The analysis starts with an initial \emph{temporal state} at the entry of the \CFG and propagates the new states all along the \CFG paths.  For each entry edge of a \BB, the input state set is the union of the output states of the preceding {\BB}s. Each input state is updated by multiplying it with the pre-computed matrix and is rebased to make a new output state. If the set of the output states differs from the original set, the successors of the current \BB are pushed into the work list. The process is repeated until finding a fix-point on all sets of input/output states.

This process may be subject to state explosion blow-up caused either by the control flow or by the timing variations i.e. the events. Using \XDD, the variability caused by events is efficiently recorded without any loss thanks to its \emph{compaction} property. Besides, the analysis at CFG level collects the set of all possible pipeline states
meaning it is also lossless according to the the variability caused by the control-flow. In turn, this means that the resulting set of vectors of {\XDD}s contains sufficient information to determine the exact temporal behavior of each \BB in all possible situations.


\section{Modeling the Shared Memory Bus}
\label{sec:shared-bus}

A frequent design in embedded microprocessors is to have the instruction and data caches sharing a common bus to the memory (or to a shared L2 cache).  
So far, our pipeline analysis required the target processor to be in-order to ensure a correct evaluation order but a~shared memory bus introduces an out-of-order behavior that raises a new difficulty: the variability created by events in the start times of FE and ME stages may change the access order to the shared bus. 
As \XG dependencies are not expressive enough to model out-of-order bus allocations,
this section proposes an extension to the pipeline analysis to manage efficiently the shared bus accesses according to the different configurations of the \emph{temporal states}. It supports the usual bus arbitration policy: first-come-first-served, with the priority given to the ME stage in case of synchronous bus accesses.



\subsection{Bus scheduling topology} 
\label{sec:bus-scheduling-topology}

Since we consider an in-order pipeline, the number of possible contention scenarios on the shared bus is limited. For instance, an instruction using the bus in the ME stage cannot contend with any subsequent instructions in the ME stage (load/store memory order is preserved). In the same fashion, the bus accesses by FE stage are performed following the \emph{Program Order}. Moreover, the \emph{Pipeline Order} ensures that a request emitted by an instruction in the FE stage acquires the bus before a request emitted by the same instruction in the ME stage. This means that the bus allocation in an in-order pipeline is almost completely in-order, with only one exception: the bus usage in the ME stage by an instruction denoted $ME_{0}$ may be delayed by a bus request in the FE stage by a subsequent instruction denoted $FE_{i|i>0}$. To simplify the notation in this section, $ME_{0}$ and $FE_{i|i>0}$ denotes as well the instructions as the \XG vertices in their respective stage. The instructions in-between are disregarded but are still accounted for in the update matrices for the \emph{temporal states}.

To sum up, $FE_i$ can delay $ME_0$ only if $FE_{i}$ is ready to enter FE stage before $ME_{0}$ is ready. In the \XG model, this situation can only happen when $FE_i$ does not depend on $ME_0$, that is, when there is no path from $ME_0$ to $FE_i$ 
\footnote{The occurrence of such situations is limited by the size of the inter-stage queues in the pipeline.}.

In the example of Table~\ref{tab:scheduling-me0}, we consider that $ME_0$ can only be delayed by $FE_1$, $FE_2$ and $FE_3$.
For a particular configuration of events, there are four possible schedules that are shown in the first column of the table. These four schedules correspond to the four possible ways to interleave $ME_0$ with $FE_i$ accesses. The actual schedule is determined by comparing the ready time of $ME_0$ ($\rho_{ME_0}$) with the ones of $FE_1$, $FE_2$ and $FE_3$ (resp. $\rho_{FE_1}$, $\rho_{FE_2}$ and $\rho_{FE_3}$): the center column shows the condition corresponding to each schedule. 
The third column gives the actual time at which $ME_0$ gets the bus with $\lambda_{BUS}$ denoting the latency to access to the bus (including the memory transaction): if $ME_0$ is the first to be ready, then it gets the bus
at time $\rho_{ME_0}$. Otherwise, $ME_0$ gets the bus at the maximum time between its ready time and the release time of the $FE_i$ contender that get the bus before.

\begin{table*}[htbp]
\centering
\begin{tabular}{|c|c|c|}
    \hline
    Schedule         & Condition                                   & Scheduling time of $ME_0$   \\
    \hline
    $ME_0, FE1, FE2, FE3$ & $\rho_{ME_0} \leq \rho_{FE_1}$                   & $\rho_{ME_0}$\\
    \hline
    $FE1, ME_0, FE2, FE3$ & $ \rho_{FE_1} < \rho_{ME_0} \leq \rho_{FE_2}$ & $max(\rho_{FE_1} + \lambda_{BUS}, \rho_{ME_0})$\\
    \hline
    $FE1, FE2, ME_0, FE3$ & $ \rho_{FE_2} < \rho_{ME_0} \leq \rho_{FE_3}$ & $max(\rho_{FE_2} + \lambda_{BUS}, \rho_{ME_0})$\\
    \hline
    $FE1, FE2, FE3, ME_0$ & $ \rho_{FE_3} < \rho_{ME_0}$                 & $max(\rho_{FE_3} + \lambda_{BUS}, \rho_{ME_0})$\\
    \hline
\end{tabular}
    \caption{Possible schedules of $ME_0$ with subsequent $FE$s.}
    \label{tab:scheduling-me0}
\end{table*}

\tikzstyle{xdd node}=[shape=circle, draw, inner sep=.4]
\begin{figure*}[htbp]
    \centering
    \tikzstyle{xdd node}=[shape=circle, draw, inner sep=.4]
    \begin{tikzpicture}[level distance=6mm,->,baseline]
            \def\off{2mm}
            \def\botline{0mm}
            \def\midline{25mm}
            \def\topline{48mm}
            \def\upline{66mm}
            \def\vshift{30mm}
            
            \node at (0*\vshift, \upline+10mm)(label_rho){$\rho_{S_i}$};
            \node[align=center] at (0.78*\vshift, \upline+13mm)(scheduleME){
                $\rho_{schedME_0} =$ \\
                $(\rho_{ME_0} \blacktriangleleft_{ME} \rho_{FE_i})$ \\
                $\oplus \rho_{rel}$
            };
            \node[align=center] at (1.83*\vshift, \upline+12mm)(rhohatME){
                $\rhohat_{ME_0} =$ \\
                $\rhohat_{ME_0} \ominus \rho_{schedME_0}$
            };
            \node[align=center] at (2.85*\vshift, \upline+12mm)(scheduleFE) {
                $\rho_{schedFE_i} =$ \\
                $\rho_{FE_i} \blacktriangleleft_{FE} \rho_{ME_0}$
            };
            \node[align=center] at (3.7*\vshift, \upline+12mm)(rhorel) {
                $\rho_{rel}$\\
            };
            \node[align=center] at (4.6*\vshift, \upline+12mm)(rhohatFE) {
                $\rhohat_{FE_i}$\\
            };
            
            \node at (-20mm, \upline-3mm)(x1){$ME_0$};
            \node at (-20mm, \topline-4mm)(x2){$FE_1$};
            \node at (-20mm, \midline-5mm)(x3){$FE_2$};

            \begin{scope}[xshift=0mm, yshift=\upline]
            \node[xdd node](e0) {\tiny $e_0$};
            \node[xdd node, below left=\off and 4mm of e0](e1) {\tiny $e_1$};
        	\node[below left=\off and 0mm of e1] (L1) { $1$ };
        	\node[below right=\off and 0mm of e1] (L3) { $3$ };
        	\node[below right=\off and 0mm of e0] (L15) { $15$ };
        	\draw[->,dashed] (e0) edge (e1);
        	\draw[->] (e0) edge (L15);
        	\draw[->,dashed] (e1) edge (L1);
        	\draw[->] (e1) edge (L3);
        	\node [above right= -1mm of e0](label) {(a)};
        	\end{scope}
        	
        	\begin{scope}[xshift=1.75*\vshift, yshift=\upline]
            \node[xdd node](L_INF) {\tiny $+\infty$};
            \node [above right= -1mm of L_INF](label) {(b)};
            \end{scope}
            
            \begin{scope}[xshift=3.7*\vshift, yshift=\upline]
            \node[xdd node](L_INF) {\tiny $-\infty$};
            \node [above right= -1mm of L_INF](label) {(c)};
            \end{scope}
            
        	\begin{scope}[xshift=0*\vshift, yshift=\topline-3mm]
        	\node[xdd node](ic1) {\tiny $ic_1$};
        	\node[below left=\off and 2mm of ic1] (L_inf){ $-\infty$};
        	\node[below right=\off and 2mm of ic1] (L2){$2$};
        	\draw[->,dashed] (ic1) edge (L_inf);
        	\draw[->] (ic1) edge (L2);
        	\node [above right=-1mm of ic1](label) {(d)};
        	\end{scope}

            \begin{scope}[xshift=0.8*\vshift, yshift=\topline]
            \node[xdd node](ic1) {\tiny $ic_1$};
            \node[xdd node, below right=\off and 2mm of ic1](e0) {\tiny $e_0$};
            \node[xdd node, below left=\off and 4mm of e0](e1) {\tiny $e_1$};
        	\node[below left=\off and 0mm of e1] (L1) { $1$ };
        	\node[below right=\off and 0mm of e1] (L_inf) { $\infty$ };
        	\draw[->,dashed] (e0) edge (e1);
        	\draw[->] (e0) edge (L_inf);
        	\draw[->,dashed] (e1) edge (L1);
        	\draw[->] (e1) edge (L_inf);
        	\draw[->,dashed] (ic1) edge (L_inf);
        	\draw[->] (ic1) edge (e0);
        	\node [above right=-1mm of ic1](label) {(e)};
        	\end{scope}
        
            \begin{scope}[xshift=1.75*\vshift, yshift=\topline]
            \node[xdd node](ic1) {\tiny $ic_1$};
            \node[xdd node, below right=\off and 2mm of ic1](e0) {\tiny $e_0$};
            \node[xdd node, below left=\off and 4mm of e0](e1) {\tiny $e_1$};
        	\node[below left=\off and 0mm of e1] (L1) { $1$ };
        	\node[below right=\off and 0mm of e1] (L_inf) {$\infty$ };
        	\draw[->,dashed] (e0) edge (e1);
        	\draw[->] (e0) edge (L_inf);
        	\draw[->,dashed] (e1) edge (L1);
        	\draw[->] (e1) edge (L_inf);
        	\draw[->,dashed] (ic1) edge (L_inf);
        	\draw[->] (ic1) edge (e0);
        	\node [above right=-1mm of ic1](label) {(f)};
        	\end{scope}

            \begin{scope}[xshift=2.85*\vshift, yshift=\topline]
            \node[xdd node](ic1){\tiny $ic_1$};
            \node[xdd node, below right=\off and 2mm of ic1](e0) {\tiny $e_0$};
            \node[xdd node, below left=\off and 2mm of e0](e1) {\tiny $e_1$};
        	\node[below left=\off and 0mm of e1] (L_inf) { $\infty$ };
        	\node[below left=\off and 0mm of ic1] (L_minf) { $-\infty$ };
        	\node[below right=\off and 0mm of e1] (L2) { $2$ };
        	\draw[->,dashed] (e0) edge (e1);
        	\draw[->] (e0) edge (L2);
        	\draw[->,dashed] (e1) edge (L_inf);
        	\draw[->] (e1) edge (L2);
        	\draw[->,dashed] (ic1) edge (L_minf);
        	\draw[->] (ic1) edge (e0);
        	\node [above right=-1mm of ic1](label) {(g)};
            \end{scope}
            
            \begin{scope}[xshift=3.65*\vshift, yshift=\topline]
            \node[xdd node](ic1){\tiny $ic_1$};
            \node[xdd node, below right=\off and 2mm of ic1](e0) {\tiny $e_0$};
            \node[xdd node, below left=\off and 2mm of e0](e1) {\tiny $e_1$};
        	\node[below left=\off and 0mm of e1] (L_inf) { $\infty$ };
        	\node[below left=\off and 0mm of ic1] (L_minf) { $-\infty$ };
        	\node[below right=\off and 0mm of e1] (L11) { $11$ };
        	\draw[->,dashed] (e0) edge (e1);
        	\draw[->] (e0) edge (L11);
        	\draw[->,dashed] (e1) edge (L_inf);
        	\draw[->] (e1) edge (L11);
        	\draw[->,dashed] (ic1) edge (L_minf);
        	\draw[->] (ic1) edge (e0);
        	\node [above right=-1mm of ic1](label) {(h)};
            \end{scope}
            
            \begin{scope}[xshift=4.6*\vshift, yshift=\topline]
            \node[xdd node](ic1){\tiny $ic_1$};
            \node[xdd node, below right=\off and 2mm of ic1](e0) {\tiny $e_0$};
            \node[xdd node, below left=\off and 2mm of e0](e1) {\tiny $e_1$};
        	\node[below left=\off and 0mm of e1] (L10) { $10$ };
        	\node[below left=\off and 0mm of ic1] (L_minf) { $-\infty$ };
        	\node[below right=\off and 0mm of e1] (L2) { $2$ };
        	\draw[->,dashed] (e0) edge (e1);
        	\draw[->] (e0) edge (L2);
        	\draw[->,dashed] (e1) edge (L10);
        	\draw[->] (e1) edge (L2);
        	\draw[->,dashed] (ic1) edge (L_minf);
        	\draw[->] (ic1) edge (e0);
        	\node [above right=-1mm of ic1](label) {(i)};
            \end{scope}

            \begin{scope}[xshift=0*\vshift, yshift=\midline]
            \def\off{3mm}
            \node[xdd node](ic1){\tiny $ic_1$};
            \node[xdd node, below right=\off and 2mm of ic1](e0) {\tiny $e_0$};
            \node[xdd node, below left=\off and 2mm of e0](e1) {\tiny $e_1$};
            \node[below left=\off*2 and 2mm of ic1] (L3) {$3$};
        	\node[below left=\off and 0mm of e1] (L19) { $19$ };
        	\node[below right=\off and 2mm of e1] (L11) { $11$ };
        	\draw[->,dashed] (e0) edge (e1);
        	\draw[->] (e0) edge (L11);
        	\draw[->,dashed] (e1) edge (L19);
        	\draw[->] (e1) edge (L11);
        	\draw[->,dashed] (ic1) edge (L3);
        	\draw[->] (ic1) edge (e0);
        	\node [above right=-1mm of ic1](label) { (j) };
        	\end{scope}
        	
            \begin{scope}[xshift=0.8*\vshift, yshift=\midline]
            \def\off{3mm}
            \node[xdd node](ic1) {\tiny $ic_1$};
            \node[xdd node, below left=\off and 1mm of ic1](e0a) {\tiny $e_0$};
            \node[xdd node, below right=\off and 1mm of ic1](e0b) {\tiny $e_0$};
            \node[xdd node, below left=\off and 0mm of e0b](e1b) {\tiny $e_1$};
            \node[xdd node, below left=\off and 0mm of e0a](e1a) {\tiny $e_1$};
            \node[below left=\off and 0mm of e1a](L1) {$1$};
            \node[below right=\off and 0mm of e1a](L3) {$3$};
        	\node[below right=\off and 0mm of e1b] (L11) { $11$ };
        	\node[below right=3*\off and 2mm of e0b] (L_inf) { $\infty$ };
        	\draw[->, dashed] (ic1) edge (e0a);
        	\draw[->] (ic1) edge (e0b);
        	\draw[->, dashed] (e0a) edge (e1a);
        	\draw[->,dashed] (e0b) edge (e1b);
        	\draw[->] (e0a) edge (L_inf);
        	\draw[->, dashed] (e1a) edge (L1);
        	\draw[->] (e1a) edge (L3);
        	\draw[->] (e0b) edge (L_inf);
        	\draw[->,dashed] (e1b) [bend right=90]edge (L_inf);
        	\draw[->] (e1b) edge (L11);
        	\node [above right=-1mm of ic1](label) {(k)};
        	\end{scope}
            
            \begin{scope}[xshift=1.75*\vshift, yshift=\midline]
            \def\off{3mm}
            \node[xdd node](ic1b){\tiny$ic_1$};
            
            \node[xdd node, below left=\off and 0mm of ic1b](e0b){\tiny $e_0$};
            \node[xdd node, below right=\off and 0mm of ic1b](e0c){\tiny $e_0$};
            
            \node[xdd node, below left=\off and 0mm of e0b](e1b){\tiny $e_1$};
            \node[xdd node, below left=\off and 0mm of e0c](e1c){\tiny $e_1$};
            
            \node[below left=\off and -1mm of e1b](L1){$1$};
            \node[below right=\off and 0mm of e1b](L3){$3$};
            \node[below right=\off and 0mm of e1c](L11){$11$};
            
            
            \draw[->, dashed] (ic1b) edge (e0b);
            \draw[->] (ic1b) edge (e0c);
            
            \draw[->, dashed] (e0b) edge (e1b);
            \draw[->, dashed] (e0c) edge (e1c);
            \node[below right= 1.5*\off and 3mm of e0c] (L_inf2){$\infty$};
            \draw[->] (e0b) edge (L_inf2);
            

            \draw[->, dashed] (e1b) edge (L1);
            \draw[->, dashed] (e1c) edge (L1);
            \draw[->] (e1b) edge (L3);
            \draw[->] (e1c) edge (L11);
            
            \draw[->] (e0c) edge (L_inf2);

            \node [above right=-1mm of ic1b](label) {(l)};
            \end{scope}
                
            \begin{scope}[xshift=2.83*\vshift, yshift=\midline-3mm]
            \def\off{3mm}
            \node[xdd node](ic1) {\tiny $ic_1$};
            \node[xdd node, below left=\off and 0mm of ic1](e0a) {\tiny $e_0$};
            \node[xdd node, below right=\off and 0mm of ic1](e0b) {\tiny $e_0$};
            \node[below right = \off and 0mm of e0b] (L11) { $11$ };
            \node[below right = \off and 0mm of e0a] (L3) { $3$ };
            \node[below left=\off and 0mm of e0a](L_inf) {$\infty$};
            \draw[->, dashed] (e0b) edge (L_inf);
            \draw[->] (e0b) edge (L11);
            \draw[->, dashed] (ic1) edge (e0a);
            \draw[->] (ic1) edge (e0b);
            \draw[->, dashed] (e0a) edge (L_inf);
            \draw[->] (e0a) edge (L3);
            \node [above right=-1mm of ic1](label) {(m)};
            \end{scope}
            
            \begin{scope}[xshift=3.65*\vshift, yshift=\midline-3mm]
            \node[xdd node](ic1b) {\tiny $ic_1$};
            
            \node[xdd node, below left=\off and 0mm of ic1b] (e0b) {\tiny$e_0$};
            \node[xdd node, below right=\off and 0mm of ic1b] (e0c) {\tiny$e_0$};
            
            \node[below left=\off and 0mm of e0b](L_inf1){$\infty$};
            \node[below right=\off and 0mm of e0b](L12){$12$};
            \node[below right=\off and 0mm of e0c](L20){$20$};
            
            
        
            \draw[->, dashed] (ic1b) edge (e0b);
            \draw[->] (ic1b) edge (e0c);
            
            \draw[->, dashed] (e0b) edge (L_inf1);
            \draw[->] (e0b) edge (L12);
            \draw[->, dashed] (e0c) edge (L_inf1);
            \draw[->] (e0c) edge (L20);
            
            
        	\node [above right=-1mm of ic1b](label) {(n)};
            \end{scope}
           
            \begin{scope}[xshift=4.6*\vshift, yshift=\midline]
            \def\off{3mm}
            \node[xdd node](ic1) {\tiny $ic_1$};
            \node[xdd node, below left=\off and 3mm of ic1](e0a) {\tiny $e_0$};
            \node[xdd node, below right=\off and 2mm of ic1](e0b) {\tiny $e_0$};
            \node[below right = \off and 1mm of e0b] (L11) { $11$ };
            \node[xdd node, below left=\off and 0mm of e0a](e1) {\tiny$e_1$};
            \node[xdd node, below left=\off and 0mm of e0b](e1b) {\tiny$e_1$};
            \node[below right = \off and 1mm of e1] (L12) { $12$ };
            \node[below right = \off and 1mm of e0a] (L3) { $3$ };
            \node[below left=\off and 2mm of e1](L10) {$10$};
            \node[below right=\off of e1b](L20) {$20$};
            
            \draw[->] (e0b) edge (L11);
            \draw[->, dashed] (e0b) edge (e1b);
            \draw[->, dashed] (ic1) edge (e0a);
            \draw[->] (ic1) edge (e0b);
            \draw[->, dashed] (e0a) edge (e1);
            \draw[->] (e0a) edge (L3);
            \draw[->, dashed] (e1) edge (L10);
            \draw[->] (e1) edge (L12);
            \draw[->, dashed] (e1b) edge (L10);
            \draw[->] (e1b) edge (L20);
            \node [above right=-1mm of ic1](label) {(o)};
            \end{scope}

            \draw[-] (-0.7*\vshift, \upline+7mm) edge (5*\vshift, \upline+7mm);
            \draw[-] (-0.45*\vshift, \upline+15mm) edge (-0.45*\vshift, \botline);
    \end{tikzpicture}
    \caption{Batch bus scheduling with $XDD$s.}
    \label{fig:xdds-during-contention}
\end{figure*}
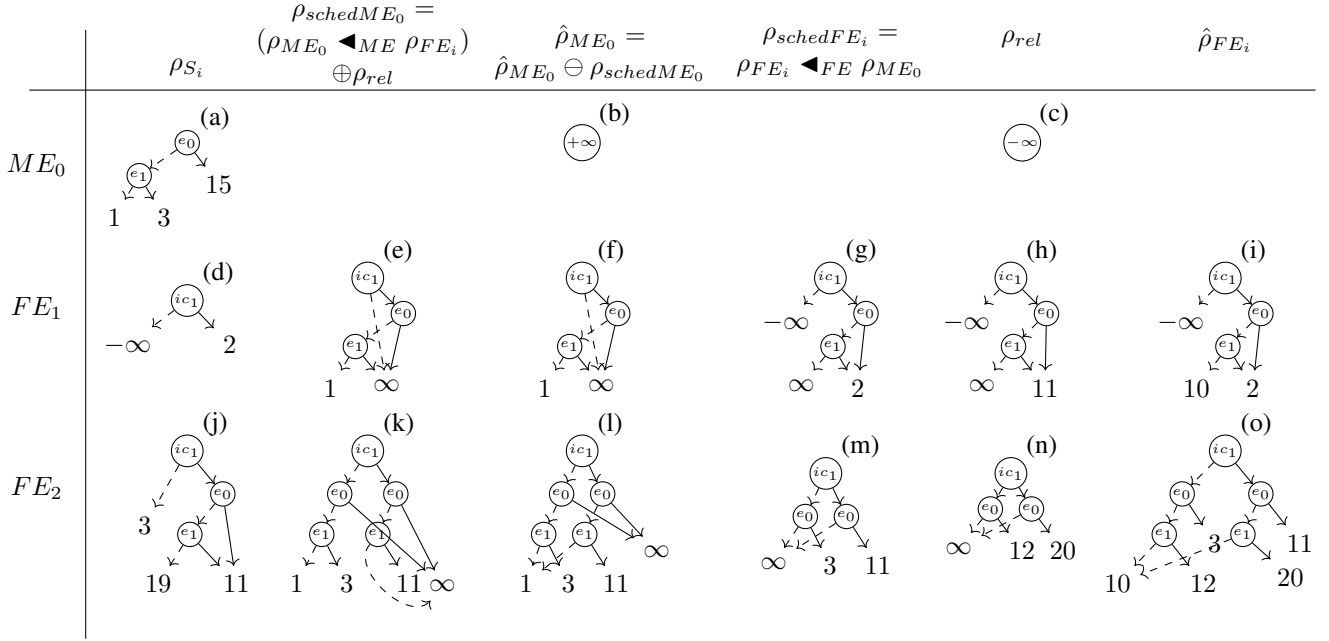

\subsection{Batch bus scheduling with \XDD}

Table~\ref{tab:scheduling-me0} shows the schedule of $ME_0$ for a fixed configuration. Yet, the times are recorded with {\XDD}s and a particular \XDD may support configurations with different schedules. Figure~\ref{fig:xdds-during-contention} shows how the bus contention scheduling presented in the previous paragraph is extended to {\XDD}s. Let us consider a lightly simpler scenario: $ME_0$ may be delayed by $FE_1$ and $FE_2$.  The instruction memory access at $FE_1$ may experiment instruction cache Hits or Misses represented by event $ic_1$. The access at $FE_2$ is classified as Always Miss, meaning that it always requests the bus. The latency of bus access is 9 cycles.

\XDD~(a) shows the ready time of $ME_0$ and (b) the initial value of $\rhohat_{ME_0}$, the scheduling time of $ME_0$ on the bus ($+\infty$ means that no access is yet scheduled). (c) shows the initial value of $\rho_{rel}$, recording the release time of the bus by $FE_i$ ($-\infty$ denotes that the bus is not used by any $FE_i$ for now). 

The ready time of $FE_1$ (d) is computed from the initial state $\vec{S}_0$ and the matrix between $ME_0$ and $FE_1$. The event $ic_1$ indicates with $-\infty$ the configuration where $FE_1$ does not use the bus  (hence it is not concerned by the contention).

$\rho_{ME_0}$~(a) and $\rho_{FE_1}$~(d) are compared using $\leftME$ to get the configurations and the time -- $\rho_{schedME_0}$~(e) where $ME_0$ takes the bus, i.e. is scheduled, before $FE_1$ ($\leftME$ is formally defined in Eq.~\ref{eq:triangle}). Other configurations are assigned $+\infty$ denoting they are not processed yet. Notice that $-\infty$ configurations in $\rho_{FE_i}$ does not allow $ME_0$ to be scheduled as subsequent $FE_{j>i}$ might allocate the bus before $ME_0$. $\rho_{schedME_0}$ is then used to update $\rhohat_{ME_0}$ using the minimum operator $\ominus$~(f). 

$\rho_{schedFE_1}$~(g), the configurations where $FE_1$ gets the bus is computed in a similar way as $\rho_{schedME_0}$ but with operator $\leftFE$ that selects the configurations of $FE$ with the strict $<$ comparison instead of $\leq$ because $ME$ stage has priority over $FE$ stage. By adding the latency of the bus ($\lambda_{BUS}$) to $\rho_{schedFE_1}$, we are able to update, using $\oplus$, the release time of the bus after $FE_1$ -- $\rho_{rel}$~(h). Finally, we compute the actual schedule of $FE_1$ -- $\rhohat_{FE_1}$ (i) which is the time of $\rho_{schedFE_1}$ if $FE_1$ is scheduled, otherwise the release time of the bus by $ME_0$ ($\rhohat_{ME_0} \otimes \lambda_{BUS}$). Now, as the actual schedule of $FE_1$ is known, the release time of the bus at $FE_1$ is computed and is used to adjust the \emph{temporal state}. By multiplying the state $\vec{S}_{FE_1}$ by the matrix $M_{FE_1-FE_2}$, we get the ready time of $FE_2$~(j).
In the second iteration, first, $\rho_{ME_0}$~(a) is compared with $\rho_{FE_2}$~(j) with the operator $\leftME$. The actual scheduling time $\rho_{schedME_0}$~(k) is computed by considering the maximum  between $\rho_{schedME_0}$ and the release time of the bus by $FE_1$ ($\rho_{rel}$) according to the third column of Table~\ref{tab:scheduling-me0}. Then, $\rhohat_{ME_0}$ is updated (l).
The schedule of $FE_2$ -- $\rho_{schedFE_2}$~(m) is computed with the operator \leftFE applied to $\rho_{FE_2}$ and $\rho_{ME_0}$ which is then used to update the release time of the bus $\rho_{rel}$~(n). The actual schedule of $FE_2$ is computed with respect to the use of the bus by $ME_0$~(o).

When the end of the sequence is reached, there are no further subsequent instructions that may contend with $ME_0$ and the remaining $+\infty$ in $\rhohat_{ME_0}$ represents configurations accessing the bus after $FE_1$ and $FE_2$. They are replaced by the maximum between the ready time of $ME_0$ and the release times of the bus by $FE_1$ and $FE_2$, $\rho_{rel}$.

Operators $\leftME$ and $\leftFE$ have a straight-forward definitions setting to $+\infty$ the configurations where $ME$, respectively $FE$, does not get the bus:

\begin{equation}
    \begin{aligned}
        & \forall \SF_{ME}, \SF_{FE} \in \XDD^2, \forall \gamma \in \Gamma, \\
        & (\SF_{ME} \leftME \SF_{FE})[\gamma] = \begin{cases}
            \SF_{ME}[\gamma]    & \text{ if  } \SF_{ME}[\gamma] \le \SF_{FE}[\gamma], \\
            +\infty             & \text{ otherwise} \\
        \end{cases} \\
        & (\SF_{FE} \leftFE \SF_{ME})[\gamma] = \begin{cases}
            \SF_{FE}[\gamma]    & \text{ if  } \SF_{FE}[\gamma] < \SF_{ME}[\gamma], \\
            +\infty             & \text{ otherwise}
        \end{cases}
    \end{aligned}
    \label{eq:triangle}
\end{equation}

All these calculations seems a bit complex but it must be kept in mind that real {\XDD}s are much more complex with much more configurations and relying on the \XDD operators allows to benefit from the {\XDD}s optimizations.
\vspace*{-3mm}
\subsection{Contention Analysis}
The contention analysis depicted in the preceding example is described more formally in this paragraph.  Basically, the pipeline analysis is extended by splitting a \BB at \emph{contention points}, the \XG node where a bus access may occur i.e. $ME$ or $FE$ stages causing cache misses. Then they are grouped in a sequence of one $ME$ access followed by zero or several $FE$ accesses, $(ME_0, FE_{0 < i \le n})$.  The instructions between \emph{contention points} are summarized by a pre-computed matrix.

Algorithm~\ref{algo:contention} is then applied to compute the possible interleaving of bus accesses for all configurations of the sequence $(ME_0, FE_{0 < i \le n})$.
Additionally, it takes as input the \emph{temporal state} $\vec{S}_0$. The result is the definitive schedule of $ME_0$ -- $\rhohat_{ME_0}$ and of $FE_i$ -- $\rhohat_{FE_i}$.

\begin{algorithm}[ht]
        \LinesNumbered \DontPrintSemicolon
        \KwIn{
            $\vec{S}_0 \in \PipeStates$, $(ME_0 ,FE_{1 \le i \le n})$ \\
        }
        \KwOut{$(\rhohat_{ME_0}, \rhohat_{FE_{1 \le i \le n}})$}
        $\rhohat_{ME_0} = \LEAF(+\infty)$ \;
        $\rho_{rel} := \LEAF(-\infty)$\;
        $\vec{S}_{FE_1} := \vec{S}_{ME_0} \cdot M_{ME_0-FE1}$ \;
        $i := 1;$ \;
        $\rho_{ME_0} := \vec{S}_0[i_\rho]$\;
        \While{ $i \le n \wedge (\exists \gamma \in \Gamma \land \rhohat_{ME_0}[\gamma] = +\infty)$}{
            \If{$FE_i.mustUseBus()$}{
                $\rho_{FE_i} := \vec{S}_{FE_i}[i_{FE}]$
            }
            \Else{
                $\rho_{FE_i} := \vec{S}_{FE_i}[i_{FE}] \otimes \NODE(ic_i, -\infty, 0) $
            }
            $\rho_{schedME_0} := (\rho_{ME_0} \leftME \rho_{FE_i}) \oplus \rho_{rel}$ \;
            $\rhohat_{ME_0} := \rhohat_{ME_0} \ominus \rho_{sched}$ \;
            \medskip
            $\rho_{schedFE_i} := \rho_{FE_i} \blacktriangleleft_{FE} \rho_{ME_0}$ \;
            $\rho_{rel} :=  \rho_{rel} \oplus (\rho_{schedFE_i} \otimes \lambda_{BUS})$ \;
            \medskip
            $\rhohat_{FE_i} := \rho_{schedFE_i} \ominus (\rhohat_{ME_0} \otimes \lambda_{ME_0})$ \;
            $\vec{S}_{FE_{i+1}} := (\vec{S}_{FE_i} \oplus [\Zero, ..., \Zero, \rhohat_{FE_i}\otimes \lambda_{BUS}]) \cdot M_{FE_i-FE_i+1} $\;
            $i = i + 1$ \;
        }
        $\rhohat_{ME_0} := \rhohat_{ME_0} \ominus (\rho_{rel} \oplus \rho_{ME_0})$ \;
        \caption{Contention computation.}
        \label{algo:contention}
\end{algorithm}

Initially, $ME_0$ is considered as not scheduled whatever the considered configuration and $\rhohat_{ME_0}$ is set to $LEAF(+\infty)$ (line 1). It will then be updated after considering the contention with each subsequent $FE_i$. When $ME_0$ does not contain $+\infty$ anymore or when all $FE_i$ has been processed, $ME_0$ schedule is complete (condition at line 6). Line~2 initializes $\rho_{rel}$ that records the release time of the bus by $FE_i$ to $-\infty$ as no $FE_i$ has been processed yet.

In line 3, the \emph{temporal state} just before $FE_1$ is computed by applying the matrix $M_{ME_0-FE_1}$ to the initial state $\vec{S}_0$; $i$ is initialized in line~4 and will range over the \emph{Contention Points}, 1 to $n$. The ready time of $ME_0$ is recorded into $\rho_{ME_0}$ at line~5.
Lines 7-10 compute the ready time of $FE_i$ if the access results always or sometimes in a Miss (according to $mustUseBus()$). The latter case is expressed by the event $ic_i$ and by adding the $\NODE(ic_i, -\infty, 0)$ to $\rho_{FE_i}$: $-\infty$ denotes the case where $ic_i$ does not arise and there is no bus access.

$\rho_{schedME_0}$, $ME_0$ configurations getting the bus before $FE_i$, is computed with \leftME at line~11 by comparing the ready time of $ME_0$ with the ones of $FE_i$. According to the last column of Table~\ref{tab:scheduling-me0}, these configurations are fixed by taking the maximum between the ready time of $ME_0$ and the release time of the bus $\rho_{rel}$. The schedule of $ME_0$ at this iteration is accumulated in the definitive schedule of $ME_0$ at line~12.
At line~13, the schedule of $FE_i$ is computed. Notice that as the ready time of $FE_i$ contains $-\infty$ to denote the case where it does not use the bus, these $-\infty$ are kept in $\rho_{shcedFE_i}$. By adding the bus latency $\lambda_{BUS}$ to $\rho_{schedFE_i}$ and then $\oplus$ with $\rho_{rel}$, the release time of the bus is only updated for configurations $\gamma$ where $FE_i$ uses and gets the~bus -- $\rho_{schedFE_i}[\gamma] \ne +\infty$ (line~14). Notice that the $+\infty$ in $\rho_{rel}$ cannot overwrite the release time of the bus by $FE_i$ because $FE_i$ cannot get the bus if any prior $FE_{j<i}$ does not get the~bus.
At line~15, the actual schedule of $FE_i$ is computed by replacing the $+\infty$ in $\rho_{schedFE_i}$ (where $FE_i$ loses contention in favor of $ME_0$) by the release time of the bus by $ME_0$. Configurations where time is $+\infty$ in $\rho_{schedFE_i}$ must not be $+\infty$ in $\rhohat_{ME_0}$ because only one of both $FE_i$ or $ME_0$ is scheduled. However, as $\rho_{schedFE_i}$ configurations different from $+\infty$  are lower than $\rhohat_{ME_0}$ (otherwise it is considered as non-scheduled), $\ominus$ can be used to implement the replacement.

At line~16, the \emph{temporal state} is updated regarding the schedules of $FE_i$, by applying $\oplus$ between the time pointer of the state vector and the release time of the bus by $FE_i$. The updated state is then multiplied by matrix $M_{FE_i-FE_{i+1}}$ to obtain the ready time of $FE_{i+1}$.
Line ~18 takes into account the $+\infty$ configurations remaining in $\rhohat_{ME_0}$ that are not already scheduled by the loop. The times assigned to these configurations are the maximum between the ready time of $ME_0$ and the bus release time by $FE_i$. Notice that $+\infty$ in the $\rhohat_{ME_0}$ may also be caused by the fact that none of $FE_i$ have used the bus: this time is recorded as $-\infty$ in $\rho_{rel}$ and is hence automatically overwritten by the ready time of $ME_0$.

\section{Experiments}
\label{sec:experiments}
The performance of the analysis strongly depends on the size of the {\XDD}s in the pipeline states and the number of pipeline states. Both characteristics are related to some inherent properties of the analyzed program and of the micro-architecture whose impact is difficult to estimate. Therefore, we experiment our analysis on realistic benchmarks that empirically provides a better understanding of the performances.
\subsection{Experiment Setup}
The pipeline used in the examples of the previous sections was chosen to improve the readability of the article.  For the experimentation, we prefer a more powerful micro-architecture with more parallelism leading to more complex \emph{temporal states}. In addition, this new pipeline allows to demonstrate the scalability of our approach.

The experimentation pipeline has 4 stages able to process 4 instructions per cycle: FE, DE, EX, CM. It fetches instructions from the main memory in the FE stage via a single level instruction cache. The FE stage is able to fetch simultaneously 4 instructions of the same memory block with a latency of 7 cycles for miss (without considering contention). The DE stage decodes the instructions and the EX stage handles all arithmetic, floating point and memory related operations in several Functional Units (FU). 4 ALUs (Arithmetic and Logic Units) are available and can be simultaneously used if no data dependencies are present. The latency of arithmetic operations is 1 cycle for addition and subtraction; 2 cycles for multiplication and 7 cycles for division. 1 FPU (Floating Point Unit) is available with latencies of 3 cycles for addition and substraction, 5 cycles for multiplication and 12 cycles for division. One MU (Memory Unit) is also available to handle memory related operations (load and store). In case of a multiple load/store operations, the memory accesses are performed in order, and if one load/store needs to use the bus, it occupies the bus until all loads/stores are completed. The latency of memory accesses is the same as for FE stage. An issue buffer at EX stage distributes the instruction to EX FUs with respect to the operations realized by the instruction. Instructions using the same FU are executed in-order in EX stage; instructions using different FUs are executed out-of-order (if no data dependencies exist). 

The instruction cache is a 16 KBits 2-way set associative LRU (Least Recent Used) cache. The data cache is a 8 KBits 2-way set associative LRU cache. Both caches have only one level and share the same bus to access the main memory. We think that this architecture is representative of mid-range processors used in real-time embedded systems.

The whole CFG analysis is implemented using the OTAWA toolbox~\cite{otawa}. Global analyses, including instruction and data cache analyses, control flow analyses etc. are provided by OTAWA.
The benchmarks are taken from the TACLe suite~\cite{taclebench} compiled for armv7 instruction set with hard floating point unit. Among 83 tasks to be analyzed, 7 of them\footnote{$pm, recursion, quicksort, huff\_enc, mpeg2, gsm\_enc, ammunition$}  failed due to limitations in OTAWA.



\vspace*{-3mm}
\subsection{Number of Temporal States}
The first experiment explores the number of \emph{temporal states} along the edges of {\BB}s over all benchmarks (representing the output of the source {\BB}s and the input of sink {\BB}s). The experimental results are shown in Figure~\ref{fig:nb-states}. The x-axis is the number of pipeline states, the y-axis is in logarithm scale and shows the number of edges for each quantity of states. The displayed statistics accumulate data from all TACLe's benchmarks. 
The risk, with our approach, is to face to a blowup in the number of states. Fortunately, the experimentation shows that most of the edges have less than 20 output states except in some rare cases where the number of states is much higher. This generally means that most of timing variations due to events are efficiently represented in the {\XDD}s of the \emph{temporal states}. As expected, the {\XDD}s successfully prevent the state explosion and keep the pipeline analysis tractable at CFG level.
The presence of some rare cases that have a lot of states is not blocking as the analysis time is reasonable in most of cases (confer Section~\ref{sec:analysis-time}).
\vspace*{-2mm}
\captionsetup{belowskip=-15pt}
\begin{figure}[h]
    \centering
    \includegraphics[width=0.49\textwidth]{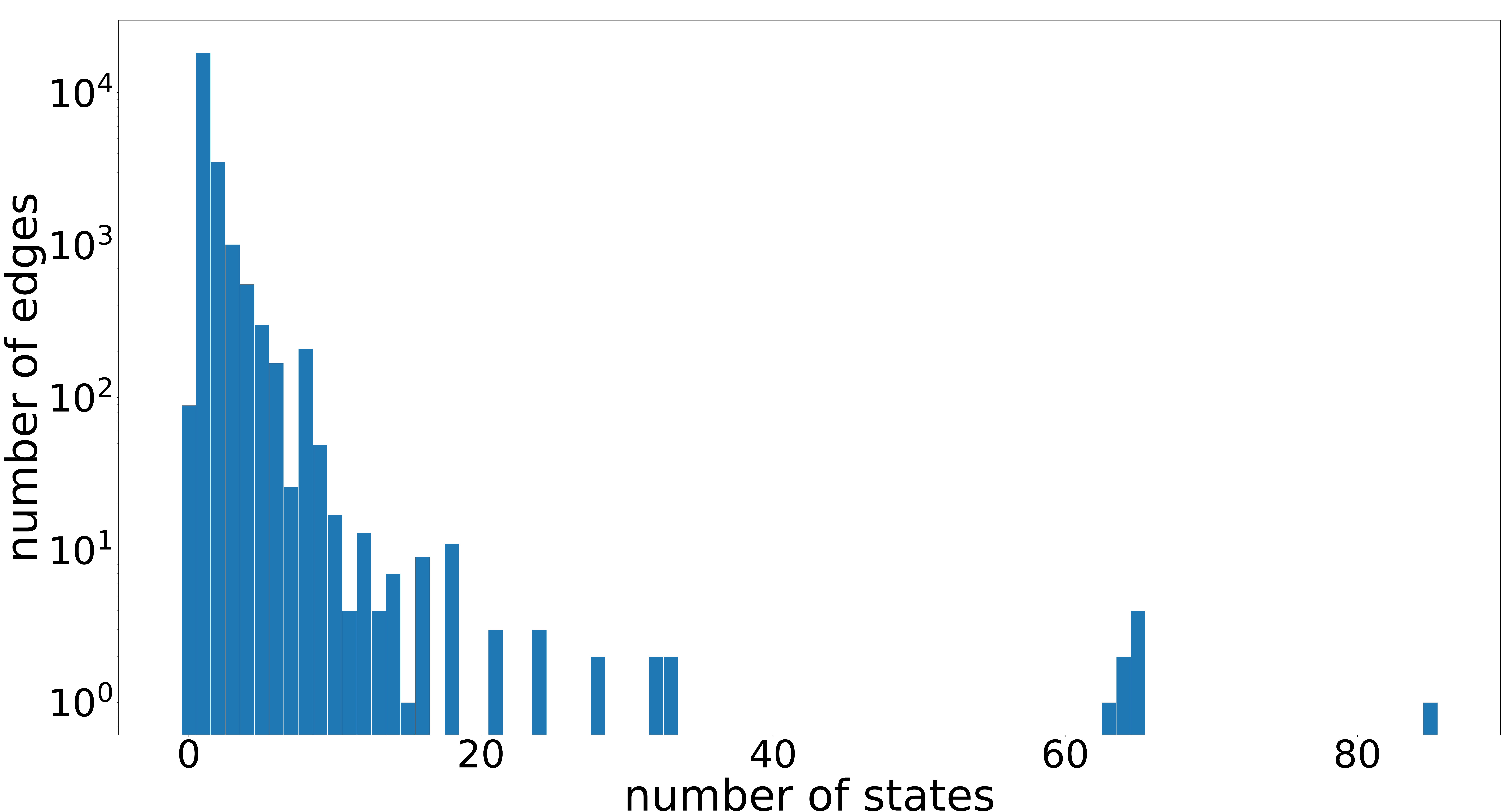}
    \caption{The distribution of number of pipeline states.}
    \label{fig:nb-states}
\end{figure}

\subsection{Events Lifetime}
The second experiment measures the lifetime of events during the analysis. The longer the lifetime of events, the larger the complexity of the analysis in terms of state number and \XDD size. In our micro-architecture, an event is created by a cache access and may  
disappear from the \XDD during the analysis, for two reasons. (a) It is absorbed by the pipeline: for example, when an instruction stalls at EX stage due to a data cache miss, next instructions may go on in the pipeline completely hiding the stalling time. This event will only stay alive in a short time window during the analysis of other instructions executed in parallel. (b) The events are \emph{stabilized} and disappear thanks to the rebasing operation. Intuitively, we assume that in most situations the events raised by an instruction only impact nearby instructions. This is demonstrated by tracking the liveness of events in the analysis.

However, the pipeline analysis is only able to provide this information at the granularity of \emph{Contention Points} because the instruction execution effect between \emph{Contention Points} is summed up by matrices. As collecting these statistics at a finer granularity would have an important adverse effect on the analysis time, we survey the liveness of events on this basis. Events are deemed as dead at \emph{Contention Points} whose \emph{temporal states} does not contain the event regarding all {\XDD}s contained in the vector. Thus, the lifetime statistics are over-estimated by the number of instructions between \emph{Contention Points}. Besides, the pipeline states are only rebased at the end of {\BB}s so the lifetime of events in the middle of {\BB}s does not consider the potential death due to rebasing. In the end, the measured lifetime in this experimentation is an over-estimation of the actual lifetime of events.

Figure~\ref{fig:events-lifetime} shows the experimental results. The x-axis is the lifetime of events (in instructions with limitations described above) and  y-axis, in logarithm scale, shows the number of events having this lifetime. These are also accumulated from the whole set of TACLe's benchmarks. The statistics show that most of events have short lifetime (below 50 instructions). We have observed a unique lifetime of 602 instructions that is not represented to keep the figure readable. It turns out that in most of situation where the lifetime is greater than 50, the events are in a \BB of a big number of instructions. In the extreme case with 602-instruction event lifetime, the involved \BB is made of 617 instructions (in benchmark \emph{md5}) and the reported lifetime is an effect of the granularity level.
Despite this very infrequent case, as most of events have short lifetime, the \emph{temporal state} size (sum of {\XDD}s sizes of the vector) stays reasonable and the analysis remains efficient.

\vspace*{-2mm}
\subsection{Analysis time}
\label{sec:analysis-time}
\begin{figure}[h]
    \centering
    \includegraphics[width=0.49\textwidth]{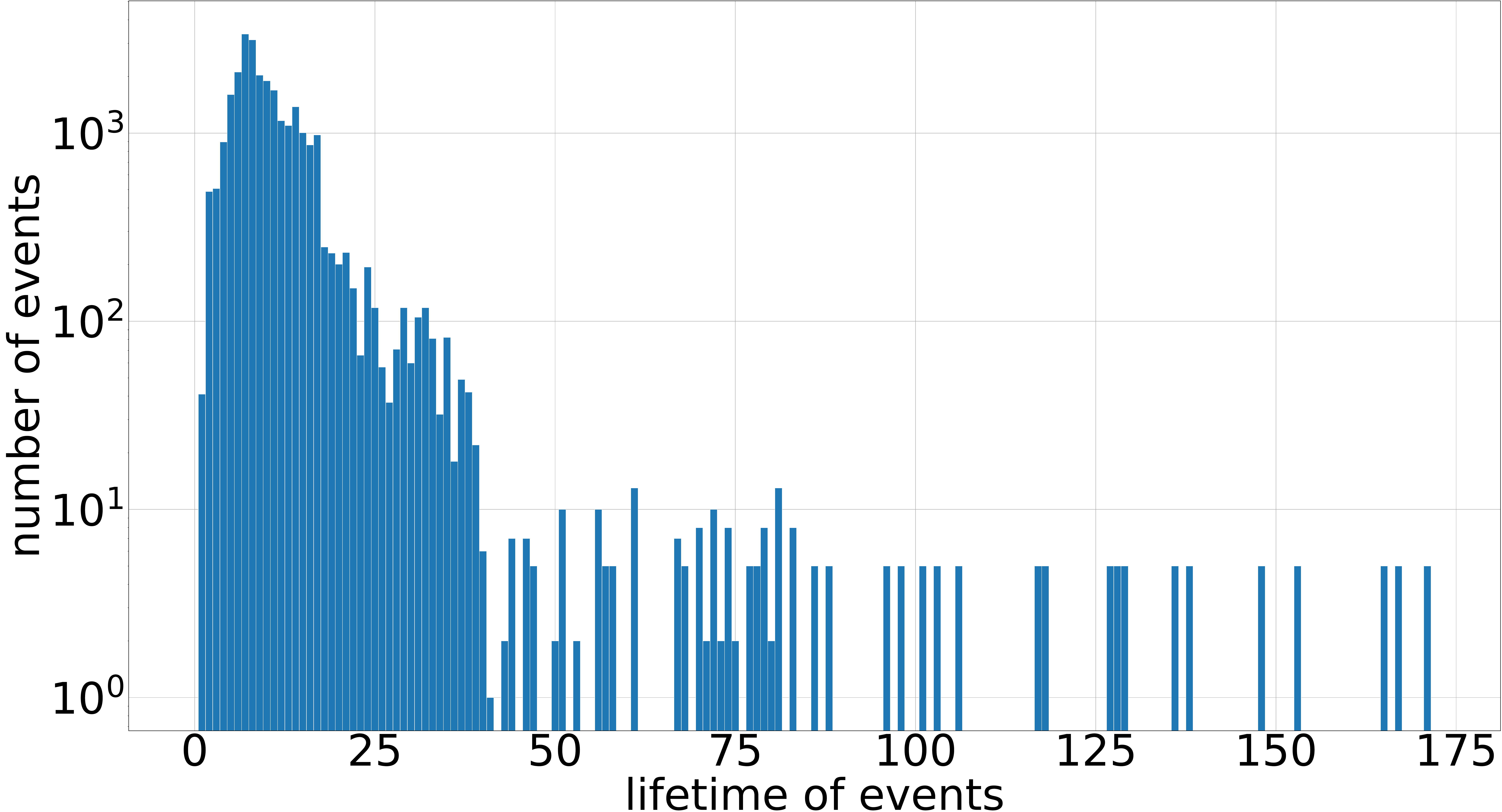}
    \caption{The distribution of events lifetime.}
    \label{fig:events-lifetime}
\end{figure}
The analysis time includes the time to pre-compute the matrices and the time of the pipeline analysis on the CFGs. The measurement is performed on a virtual machine running on a cloud server with 8GB RAM and 4-core Intel Broadwell processors. Only 2 cores are occupied simultaneously to run the benchmarks. We have also measured the analysis time without pre-computing the matrices (with a timeout of 1 hour) in order to clarify the advantage of that optimization. The results are shown in Figure~\ref{fig:analysis-time}. The x-axis shows the benchmarks and the y-axis provides the analysis time in seconds with a logarithmic scale. The analysis time with matrices is recorded as green bars. At worst, the analysis with matrices finishes in 553s (9m13s). In most cases, the analysis finishes in about $1-20s$. In contrast, the analysis without matrix has both memory usage and speed issues as shown by the red (crashed because out of 8G RAM) and yellow bars (timeout after 1~hour). For those finishing within 1 hour (blue bars), matrix optimization brings 217\% speed-up in average\footnote{average speed up = sum of the time of all benchmarks without matrix divided by sum of time with matrix.}. The rare cases where the analysis without matrix is faster are simple benchmarks where the cost of computing the matrices is not compensated by the speed-up. This result reveals that the pre-computation of matrices effectively reduces intermediate redundant computations, which enhances the analysis performance in terms of speed and memory usage. Considering the presence of some exceptional cases, large number of \emph{temporal states}, long lifetime events, we think that the analysis is still able to handle them and finishes in a reasonable time. Moreover, we observed that the industry-like applications encompassed in TACLe benchmarks -- Rosace, Debbie and Papabench, are analyzed in short times -- respectively 15s, 10mn, 15mn summing the times of all tasks composing them. Therefore, our approach could be used in industrial real-time applications (for example, Airbus requires at most 48H between the detection of a bug and the distribution of the fix, including temporal verification).
\begin{figure*}[h]
    \centering
    \includegraphics[width=0.95\textwidth]{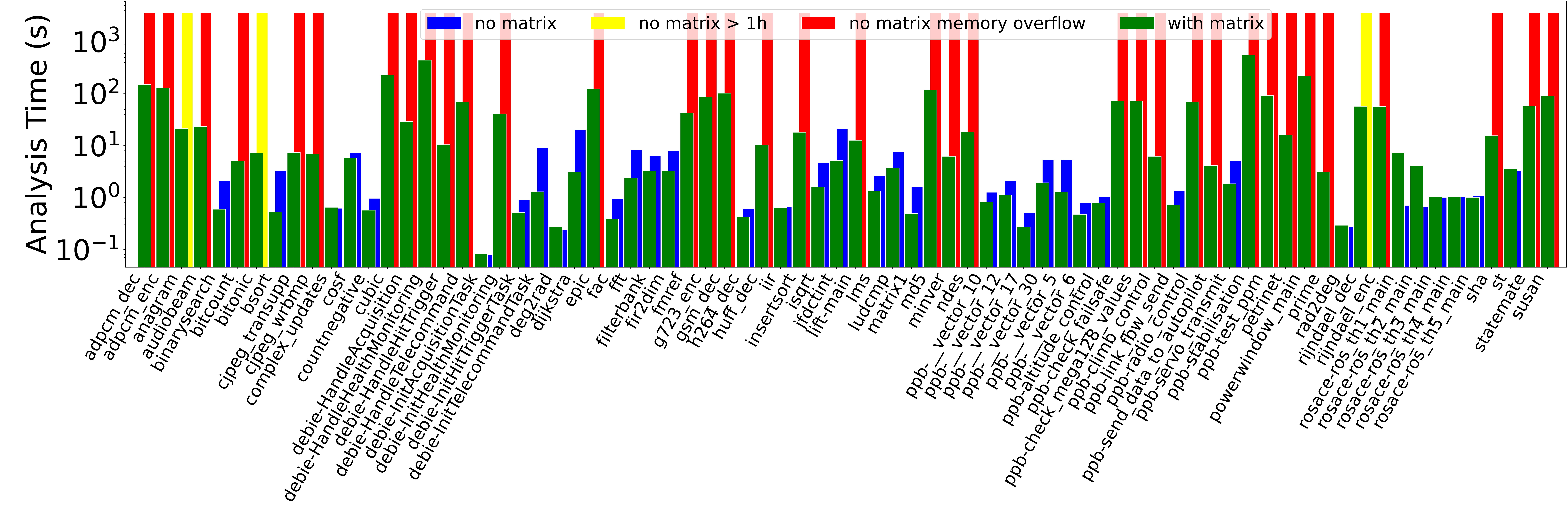}
    \vspace*{-0.8cm}
    \caption{The analysis time.}
    \label{fig:analysis-time}
\end{figure*}

\vspace*{-2mm}
\section{Related Works}
\label{sec:related}

The pipeline model used in the aiT WCET analyzer is maybe the most successful and the most used pipeline analysis\cite{schneider_pipeline_1999, thesing_safe_2004}. They define the state
by the time left in each resource of the pipeline and update it at the granularity of the processor cycle. Based on Abstract Interpretation framework\cite{abstract-interpretation0}, according to our knowledge, they use power set domain to keep the set of possible pipeline states. Therefore, this model also suffers from combinatorial complexity caused by the presence of timing anomalies as it has to keep all possibles states. They have proposed several approaches to reduce the complexity. (a) Although the literature provides very few details about that, close states seem to be joined to form abstract states at the cost of loss of precision, (b) in~\cite{wilhelm2011symbolic}, they show how to use Binary Decision Diagram to compress the state machine representation of their analysis system. Reineke et al. in~\cite{reineke_sound_2009} defines sufficient condition to drop \emph{not-worst} cases in order to reduce the number of states. This work has been extended later in \cite{hahn_toward_2015, hahn_design_2020} that provides a theoretical basis to design strict-in-order pipeline where timing anomalies are proven to not occur~\cite{wenzel_principles_2005, engblom_processor_2002-1}, thus allowing to more easily drop  non-worst case states. However, for now, the price to pay for the design of timing-anomaly-free pipelines is still a significant loss of performance.

Model checking can also be used for WCET analyses \cite{cassez-model-checking-timed-games,tic-model-checking}. With the help of mature theories and tools of model checking, the solving procedure itself is well optimized and is independent of the timed model of the target program and the micro-architecture which eases the design of modular and flexible analyses\cite{cassez-model-checking, dalsgaard-model-checking}. 
However, in general, model checking-based analyses have to completely explore the domain of traces of the program and the domain of program inputs. On the one hand, this provides tight WCETs without over-estimation, and precise information about the worst execution pattern. On the other hand, the large search domain combined with the micro-architecture model complexity questions the scalability of these analyses for complex program and architectures.

Another approach to pipeline analysis is the \emph{Execution Graph} proposed by Li et al \cite{xianfeng_li_modeling_2004}, close to what is presented in Section~\ref{sec:XG}. They analyze the WCET at the scope of {\BB}s and calculate the worst execution context. To support timing variation, the \XG computational model uses intervals to representing minimum and maximum times. The contention between instructions is considered by checking the intersection of time intervals. If a contention occurs, the interval is extended accordingly. The \XG solving algorithm repeats the computation until a fix-point is reached. However, in the presence of lots of events the interval representation tends to trigger a chain reaction: the imprecision due to the interval representation create contentions that are actually impossible which extends the interval and involves more impossible contentions. Moreover, with respect to the micro-architecture, making precise assumption on the worst execution context is not always trivial. Another \XG based approach is proposed by Rochange et al. in~\cite{rochange_context-parameterized_2009} that computes the execution time of {\BB}s for each combination of events what makes the algorithm to tend toward combinatorial complexity. In addition, the contention analysis requires to examine all cases leading to an  exponential complexity.

\vspace*{-3mm}
\section{Conclusion}
\label{sec:conclusion}
In this paper, we have formally defined a state representation useful for the pipeline time analysis. It is derived from the \XG model but the times are replaced by {\XDD}s that efficiently represent time variation caused by the raise of events. An \XDD is a data structure working as a lossless compression of a map between the event configurations and the execution time. {\XDD}'s advantage is its ability to compact the time variation 
compared to the use of power set map. Moreover, we represent the pipeline state as a vector of {\XDD}s. This simplifies the design of pipeline analysis at \CFG level and allows to leverage the algebraic properties of {\XDD}s to represent the analysis of instruction sequences as matrices multiplications. These matrices can be statically determined before the analysis which significantly speeds up the analysis.
Together with rebasing and generation number, the presented analysis enables the tracking of exact timing behavior over the CFG of program.

Secondly, we extend this analysis to support the shared bus between fetch and memory pipeline stages. The shared bus is dynamically allocated and may experiment out-of-order access. \emph{Temporal states} obtained so far are used to track precisely the bus accesses schedule. Based on the survey of the topology of the bus usage by instructions, we designed a contention analysis to support bus access times in the \emph{temporal states}. The contention analysis needs only to be invoked upon \emph{contention points} while instructions in-between are summarized by the aforementioned matrices.


The experimentation has been conducted on TACLe's benchmarks. The measurement of the number of pipeline states per edge in the \CFG showed that our approach is able to efficiently represent the timing variations. Then, we produce a rough (but conservative) evaluation of the event lifetime that shows that the effect of events are generally short term in the \CFG.
Some exceptions are observed (edges with a lot of states or events with long lifetime), but they are not problematic as they are very infrequent. The analysis time shows that our analysis is very efficient and suggests that it could be used for industrial applications. 

As future works, we could benefit from the exact tracking of the \emph{temporal states} to more precisely qualify the effects of timing variations in different micro-architectures. This could be used to eventually qualify good or bad micro-architecture design in terms of predictability. This may help to find better compromises in the design of predictable pipelines, which could also alleviate over-stringent constraints on the pipeline such as strict-in-order execution, that often limits the performance of the processor.
We plan to extend our approach
to all out-of-order resources. Although our operators and matrices calculation are correct whatever the out-of-order resource, the modeling of the interleaving of resource acquisitions might not scale.


\vspace*{-3mm}
\bibliographystyle{IEEEtran}
\bibliography{IEEEabrv, main}

\end{document}